%% file: loops_CoRR.tex
\documentclass{tlp}
\usepackage[utf8x]{inputenc}
\usepackage{amsmath}
\usepackage{amssymb}
\usepackage{graphicx}
\usepackage{url}
\usepackage[x11names, rgb]{xcolor}
\usepackage{tikz}
\usetikzlibrary{snakes,arrows,shapes}

\include{dvasp}

\include{dvfuzzy}

\include{general}
\include{jjgeneral}

\title{Reducing Fuzzy Answer Set Programming to Model Finding in Fuzzy Logics}
\author[Jeroen Janssen, Steven Schockaert, Dirk Vermeir and Martine De Cock]{JEROEN JANSSEN\thanks{Funded by a joint Research Foundation-Flanders (FWO) project}\\Dept.~of Computer Science, Vrije Universiteit Brussel\\Pleinlaan 2, 1050 Brussels, Belgium\\
       \email{jeroen.janssen@vub.ac.be}
       \and STEVEN SCHOCKAERT\thanks{Postdoctoral fellow of the Research Foundation-Flanders (FWO)}\\Dept.~of Applied Mathematics and Computer Science, Universiteit Gent\\Krijgslaan 281, 9000 Ghent, Belgium\\
	\email{steven.schockaert@ugent.be}
       \and DIRK VERMEIR\\Dept.~of Computer Science, Vrije Universiteit Brussel\\Pleinlaan 2, 1050 Brussels, Belgium\\
	\email{dirk.vermeir@vub.ac.be}
       \and MARTINE DE COCK\thanks{On leave of absence from Ghent University}\\Institute of Technology, University of Washington\\1900 Commerce Street, WA-98402 Tacoma, USA\\
	\email{mdecock@u.washington.edu}}

  \submitted{August 16, 2010}
  \revised{January 17, 2011}
  \accepted{March 26, 2011}

\begin{document}

\maketitle

\begin{abstract}
 In recent years answer set programming has been extended to deal with multi-valued predicates. The resulting formalisms
allows for the modeling of continuous problems as elegantly as ASP allows for the modeling of discrete problems, by combining the
stable model semantics underlying ASP with fuzzy logics. However, contrary to the case of classical ASP where many
efficient solvers have been constructed, to date there is no efficient fuzzy answer set programming solver. A well-known
technique for classical ASP consists of translating an ASP program $P$ to a propositional theory whose models exactly
correspond to the answer sets of $P$. In this paper, we show how this idea can be extended to fuzzy ASP, paving the way
to implement efficient fuzzy ASP solvers that can take advantage of existing fuzzy logic reasoners. To appear
in Theory and Practice of Logic Programming (TPLP).
\end{abstract}


\section{Introduction}

Answer Set Programming (ASP), see e.g.~\cite{BaralBook} is a form of non-monotonic reasoning based on the stable model semantics for logic programming \cite{gelfondl88}. Intuitively, in answer set programming one writes a set of rules (the program) such that certain minimal models (the answer sets) of this program correspond to solutions of the problem of interest.

In recent work, logic programming has been extended to handle many different facets of imperfect information. Most
notably are the probabilistic
\cite{baral:plog,damasio-hybridprobabilistic,Fuhr2000,lukasiewicz-probabilistic,Lukasiewicz-DisjunctiveProbabilisticLP,NgSubrahmanian-1993a,NgSubrahmanian-1994a,Straccia:reasoningweb} and possibilistic
\cite{Alsinet:possibilistic,Kim:UAI2010,nicolas:possibilistic-stable-models,Nicolas:AMAI2006}
extensions to handle uncertainty, the fuzzy extensions
\cite{Cao2000,IshizukaMitsuru-1985a,Luka06,LukaStraccia07,LukaStraccia-TopkRetrieval,MadridOjeda-Aciego-2008a,MadridAciego2009,Saad-2009,Straccia:reasoningweb,FASP:amai,Vojtas01,Wagner-1998} which allow to encode the intensity to
which the predicates are satisfied, and, more generally, many-valued extensions
\cite{damasio:sortedmultiadjoint,DamasioMO07,damasio-antitonic,DamasioViegasPereira2001,DamasioPereira-embeddings,vanEmden1986,Fitting1991,KiferLi-1988a,kifer92theory,Lakshmanan1994,LakshmananSadri-1994,LakshmananSadri-1997,LakshmananSadri-2001,PhDLakshmanan,LoyerStraccia-2002a,LoyerStraccia-2003,NerodeRemmelSubrahmanian1997,Shapiro1983,Straccia05queryanswering,straccia-annotated,straccia:fixedpoint,Subrahmanian1994}.
In this paper we focus on a fuzzy extension of ASP, called fuzzy answer set programming (FASP), which combines the
stable model semantics for logic programming with fuzzy logics. More generally, FASP provides a semantics for logic
programs in which the truth of predicates (or propositions) may be graded. Such grades may mean different things in
different applications, but often they are related to the intensity to which a given property is satisfied. From an
application point of view, this is important because it allows to describe continuous phenomena in a logical setting.
Thus a formalism is obtained in which problems with continuous domains can be modeled with the same ease by which
discrete problems can be modeled in classical ASP. 

In recent years, efficient solvers for classical ASP have been developed. Some of these are based on the DPLL algorithm \cite{dpll} such as Smodels \cite{smodels} and DLV \cite{dlv}, others use ideas from SAT solving such as clasp \cite{clasp}, while still others directly use SAT solvers to find answer sets, e.g.~ASSAT \cite{assat-linzhao}, cmodels \cite{cmodels}, and pbmodels \cite{pbmodels}. The SAT based approaches have been shown to be fast, and have the advantage that they can use the high number of efficient SAT solvers that have been released in recent years. The DPLL based solvers have the advantage that they allow a flexible modeling language, since they are not restricted to what can directly and efficiently be translated to SAT, and that they can be optimized for specific types of programs.

Probabilistic ASP can be reduced to classical SAT \cite{Saad:ECSQARU2009}, allowing implementations using regular SAT solvers. Likewise, possibilistic ASP can be reduced to classical ASP \cite{Nicolas:AMAI2006}, which means ASP solvers can be used for solving possibilistic ASP programs. 

In the case of fuzzy ASP programs with a finite number of truth values, it has been shown in \cite{fasp1} that FASP can be solved using regular ASP solvers. Unfortunately, to date, no fuzzy ASP solvers or solving methods have been constructed for programs with infinitely many truth values. Our goal in this paper is to take a first step towards creating such efficient solvers by showing how the idea of translating ASP programs to SAT instances can be generalized to fuzzy answer set programs. In this way we can create fuzzy answer set solvers that use existing techniques for solving fuzzy satisfiability problems, e.g.~based on mixed-integer programming or other forms of mathematical programming.
Specifically, we focus on the ASSAT approach introduced in \cite{assat-linzhao}. While translating ASP to SAT is straightforward when ASP programs do not contain cyclic dependencies, called loops, careful attention is needed to correctly cover the important case of programs with loops. The solution presented by ASSAT is based on constructing particular propositional formulas for any loop in the program. In this paper, we pursue a similar strategy where \bemph{fuzzy loop formulas} are used to correctly deal with loops. Our main contributions can be summarized as follows:

\begin{enumerate}
 \item We define the completion of a fuzzy answer set program in the sense of \cite{FASP:amai}, and show that the answer sets of FASP programs without loops are exactly the models of its completion.
 \item By generalizing the loop formulas from \cite{assat-linzhao}, we then show how the answer sets of arbitrary FASP programs can be found, provided that the fuzzy logical connectives are t-norms. We furthermore show how the ASSAT procedure, which attempts to overcome the problem with an exponential number of loops, can be generalized to the fuzzy case.
\end{enumerate}

We furthermore show that the FASP semantics in terms of unfounded sets \cite{fasp1} coincide with the FASP semantics in terms of fixpoints (see e.g.~\cite{Luka06}). This is necessary because the development of loop formulas can more easily be done using the unfounded semantics, while the generalization of the ASSAT procedure is based on the fixpoint semantics.

The structure of the paper is as follows. Section~\ref{sec:prelims} recalls the basic fuzzy logic operators and Section~\ref{sec:FASP}  recalls the FASP framework from \cite{FASP:amai}. In Section~\ref{sec:FASP} we furthermore show that the unfounded semantics and fixpoint semantics of FASP coincide. Next, we define the completion of a FASP program in Section~\ref{sec:completion} and discuss the problems that occur in programs with loops. Section~\ref{sec:loopelimination} then shows how these problems can be solved by adding loop formulas to the completion. We illustrate our approach on the problem of placing a set of ATM machines on the roads connecting a set of cities  such that each city has an ATM machine nearby in Section~\ref{sec:example}. The reason for restricting our approach to t-norms is discussed in Section~\ref{sec:restrict}. Afterwards, in Section~\ref{sec:related}, we give an overview of related work and then present the conclusions in Section~\ref{sec:conclusion}.

A preliminary version of this paper appeared in \cite{FASP-iclp08}. This paper extends our earlier work by adding proofs,
a detailed use case and a discussion on the problems that occur when programs are allowed to contain t-conorms in the body. Furthermore we improved the presentation by removing the aggregation-based approach that was used in the aforementioned work.

\section{Preliminaries}\label{sec:prelims}


In general, fuzzy logics are logics whose semantics are defined in terms of variables that can take a truth value from the unit interval $[0,1]$ instead of only the values $0$ (false) and $1$ (true). Different ways exist to extend the classical logic connectives, leading to different logics with different tautologies and axiomatizations \cite{Hajek98,novak:1999}. We briefly recall the most important concepts related to fuzzy logic connectives.

A \bemph{negator} is a decreasing $[0,1] \to [0,1]$ mapping $\mathcal{N}$ satisfying $\mathcal{N}(0) = 1$ and $\mathcal{N}(1) = 0$. A negator is called \bemph{involutive} iff $\Forall{x \in [0,1]}{\mathcal{N}(\mathcal{N}(x)) = x}$. 

A \bemph{triangular norm} (t-norm) is an increasing, commutative and associative $[0,1]^2 \to [0,1]$ operator $\pretnorm$ satisfying $\Forall{x \in [0,1]}{\mathcal{T}(1,x) = x}$. Intuitively, this operator corresponds to logical conjunction. In this paper, we restrict ourselves to left-continuous t-norms. As the most commonly used t-norms obey this restriction, this poses no practical constraint. 

A \bemph{triangular conorm} (t-conorm) is an increasing, commutative and associative $[0,1]^2 \to [0,1]$ operator $\mathcal{S}$ satisfying $\Forall{x \in [0,1]}{\mathcal{S}(0,x) = x}$. Intuitively, it corresponds to logical disjunction. 

An \bemph{implicator} $\mathcal{I}$ is a $[0,1]^2 \to [0,1]$ operator that is decreasing in its first and increasing in its second argument, and satisfies $\prefimp(0,0) = \prefimp(0,1) = \prefimp(1,1) = 1$, $\prefimp(1,0) = 0$ and  $\Forall{x \in [0,1]}{\mathcal{I}(x,1) = x}$. Every left-continuous t-norm induces a \bemph{residual implicator} defined by $\mathcal{I}(x,y) = \sup \{ \lambda \in [0,1] \mid \mathcal{T}(x,\lambda) \leq y \}$. Furthermore, a left-continuous t-norm $\mathcal{T}$ and its residual implicator $\mathcal{I}$ satisfy the \bemph{residuation principle}, i.e.~for $x,y,z$ in $[0,1]$ we have that $\mathcal{T}(x,y) \leq z$ iff $x \leq \mathcal{I}(y,z)$. For any left-continuous t-norm $\pretnorm$, its residual implicator $\prefimp$ satisfies
 \begin{equation}
  \prefimp(x,y) = 1 \textnormal{~iff~} x \leq y\label{eq:resfimpprop}
 \end{equation}
For a given implicator $\prefimp$ its induced negator is the operator $\mathcal{N}$ defined by $\fneg{}{x} = \prefimp(x,0)$. We summarize some common t-norms, t-conorms, residual implicators, and induced negators in Tables~\ref{tab:fuzzy-operators} and \ref{tab:fuzzy-operators-2}.

\begin{table}
\[\begin{array}{ll}
\hline
\hline
\textnormal{t-norm} & \textnormal{t-conorm}\\
\hline
\pretnorm_m(x,y) = \min(x,y) & \pretconorm_m(x,y) = \max(x,y) \\
\pretnorm_l(x,y) = \max(0,x+y-1) & \pretconorm_l(x,y) = \min(x+y,1) \\
\pretnorm_p(x,y) = x \cdot y & \pretconorm_p(x,y) = x + y - x \cdot y \\
\hline
\hline
\end{array}\]
\caption{Common fuzzy t-norms and t-conorms over $([0,1],\leq)$}
\label{tab:fuzzy-operators}
\end{table}


\begin{table}
\[\begin{array}{cll}
\hline
\hline
\textnormal{t-norm} & \textnormal{residual implicator} & \textnormal{induced negator}\\
\hline
  \pretnorm_m
  & \prefimp_m(x,y) = \begin{cases}
		     y & \mbox{~if~} x > y\\
		     1 & \mbox{~otherwise}
                    \end{cases}
  & \mneg{x} = \begin{cases}
                           0 & \mbox{~if~} x > 0\\
			   1 & \mbox{~otherwise}
                          \end{cases}\\
 \pretnorm_l
 & \prefimp_l(x,y) = \min(1,1-x+y)
 & \lneg{x} = 1-x\\ 
 \pretnorm_p
 & \prefimp_p(x,y) = \begin{cases}
                      y/x & \mbox{~if~} x > y\\
		      1 & \mbox{~otherwise}
                     \end{cases}
 & \pneg{x} = \mneg{x}\\
\hline
\hline
\end{array}\]
\caption{Common residual pairs and induced negators over $([0,1],\leq)$}
\label{tab:fuzzy-operators-2}
\end{table}

\bemph{The biresiduum} of a residual implicator $\prefimp$ is denoted as $\feq$, and defined by $x \feq y = \mathcal{T}(\mathcal{I}(x,y),\mathcal{I}(y,x))$. Note that due to (\ref{eq:resfimpprop}) it follows that $x \feq y$ is always equal to either $\mathcal{I}(x,y)$ or $\mathcal{I}(y,x)$. We denote the particular choice of t-norm and implicator using a subscript, when it is not clear from the context, as in $x \feq_m y = \pretnorm_m(\prefimp_m(x,y),\prefimp_m(y,x))$.


A \bemph{fuzzy set} $A$ in a universe $X$ is an $X \to [0,1]$ mapping. For $x \in X$ we call $A(x)$ the \bemph{membership degree} of $x$ in $A$. For convenience we denote with $A = \{a_1^{k_1},\ldots,a_n^{k_n}\}$ the fuzzy set $A$ satisfying $A(a_i) = k_i$ for $1 \leq i \leq n$ and $A(a) = 0$ for $a \not\in \{a_1,\ldots,a_n\}$. We use $\Fuzzy{X}$ to denote the universe of all fuzzy sets in $X$. The \bemph{support} of a fuzzy set $A$ is defined by $\supp{A} = \{ x \in X \mid A(x) > 0 \}$. Inclusion of fuzzy sets in the sense of Zadeh is defined as $A \subseteq B$ iff $\Forall{x \in X}{A(x) \leq B(x)}$.
 Last, in this paper we will write the \bemph{difference} $A \fsetminus B$ of two fuzzy sets to denote the fuzzy set defined by $(A \fsetminus B)(x) = \max(0,A(x)-B(x))$.

A \bemph{signature} is a tuple $\mathbb{S} = \langle \mathbb{A},\mathbb{T},\mathbb{C},\mathbb{I},\mathbb{N} \rangle$, with $\mathbb{A}$ a set of atoms (i.e.~propositional letters), $\mathbb{T}$ a set of t-norms, $\mathbb{C}$ a set of t-conorms, $\mathbb{I}$ a set of implicators, and $\mathbb{N}$ a set of negators. Additionally, we demand that $\max$ and $\min$ are in the signature. A \bemph{fuzzy formula} over a signature $\mathbb{S}$ then is either an atom, a value from $[0,1]$, or the application of a t-norm or t-conorm from $\mathbb{T}$, resp.~$\mathbb{C}$, to two formulas, the application of an implicator from $\mathbb{I}$ to two formulas, or the application of a negator from $\mathbb{N}$ to a single formula. A \bemph{fuzzy theory} over a signature $\mathbb{S}$ is a set of fuzzy formulas over $\mathbb{S}$. An \bemph{interpretation} $I$ over a signature $\mathbb{S}$ is a mapping from $\mathbb{A}$ to $[0,1]$. It is extended to fuzzy formulas in a straightforward way, i.e.~if $F$ and $G$ are fuzzy formulas, then $I(\pretnorm(F,G)) = \pretnorm(I(F),I(G))$ (with $\pretnorm \in \mathbb{T}$), $I(\pretconorm(F,G) = \pretconorm(I(F),I(G))$ (with $\pretconorm \in \mathbb{C}$), $I(\prefimp(F,G)) = \prefimp(I(F),I(G))$ (with $\prefimp \in \mathbb{I}$), and $I(\fneg{}{F}) = \fneg{}{I(F)}$ (with $\mathcal{N} \in \mathbb{N}$). An interpretation $M$ is a \bemph{model} of a fuzzy formula $F$, denoted $M \models F$, iff $M(F) = 1$. An interpretation $M$ is a model of a fuzzy theory $\Theta$, denoted $M \models \Theta$, iff for each $F \in \Theta$ we have that $M \models F$.

A particular signature leads to a particular fuzzy logic \cite{Hajek98}. For example, $\mathbb{S} = \langle \mathbb{A},$ $\{\pretnorm_m\},$ $\{\pretconorm_m\},$ $\{\prefimp_m\},$ $\{\mathcal{N}_m\} \rangle$ gives rise to G\"odel logic, $\mathbb{S} = \langle \mathbb{A},$ $\{\pretnorm_l,\pretnorm_m\},$ $\{\pretconorm_l,\pretconorm_m\},$ $\{\prefimp_l\},$ $\{\mathcal{N}_l\} \rangle$ gives rise to \L ukasiewicz logic, and $\mathbb{S} = \langle \mathbb{A},$ $\{\pretnorm_p,\pretnorm_m\},$ $\{\pretconorm_p,\pretconorm_m\},$ $\{\prefimp_p\},$ $\{\mathcal{N}_p\} \rangle$ gives rise to product logic. For example, \L ukasiewicz logic is generally considered to be closest in spirit to classical logic, in the sense that many of its important properties are preserved. Another important advantage of \L ukasiewicz logic is that the implicator is continuous, which is not the case for G\"odel or product logic. Reasoning in this logic can be done using mixed integer programming, whereas reasoning in G\"odel logic can be done with the help of boolean SAT solvers. As in the boolean case, satisfiability checking in these three particular logics is NP-complete.

Last, we denote the infimum, resp.~supremum of two elements of $[0,1]$ as $a \glb b$, resp.~$a \lub b$.


\section{Fuzzy Answer Set Programming}\label{sec:FASP}


Over the years many different fuzzy answer set programming formalisms have been developed \cite{Luka06,LukaStraccia07,MadridOjeda-Aciego-2008a,MadridAciego2009,Saad-2009,Straccia:reasoningweb,FASP:amai}. Most of these base their semantics on fixpoints or minimal models, in combination with a reduct operation. The approaches described in \cite{loyer:epistemic,FASP:amai}, however, are constructed from a generalization of unfounded sets. As the development of loop formulas can be done more elegantly when starting from unfounded sets, and we can show that the fixpoint semantics are equivalent to the semantics proposed in \cite{FASP:amai}, we will base our semantics on the latter framework. However, because the generalization of the ASSAT procedure from \cite{assat-linzhao} is based on fixpoint semantics, in this section we also show novel links between the unfounded and fixpoint semantics that ensure the correctness of our generalized procedure. 
First, we recall the main definitions from \cite{FASP:amai}.

 A \bemph{literal}\footnote{As usual, we assume that programs have already been grounded.} is either an atom $a$ or a \bemph{constant} from $[0,1]$. An \bemph{extended literal} is either a literal (called a \bemph{positive extended literal}) or an expression of the form $\fneg{}{a}$ (called a \bemph{negative extended literal}), with $a$ an atom and $\mathcal{N}$ an arbitrary negator. A \bemph{rule} $r$ is of the form 
  $$r: a \gets \pretnorm(b_1,\ldots,b_n)$$
 where $n > 0$, $a$ is a literal, $\{b_1,\ldots,b_n\}$ is a set of extended literals, $\pretnorm$ is an arbitrary t-norm, and $r$ is a rule label. Furthermore, for convenience, we define $\pretnorm(b) = b$, and define $\pretnorm(b_1,\ldots,b_n)$ as $\pretnorm(b_1,\pretnorm(b_2,\ldots))$. The literal $a$ is called the \bemph{head}, $\head{r}$, of $r$, while the set $\{b_1,\ldots,b_n\}$ is called the \bemph{body}, $\body{r}$, of $r$. We use $\poslit{\body{r}}$ to denote the set of positive extended literals from the set $\body{r}$.
 Given a rule $r$ we denote the t-norm used in its body as $\bodyfand{r}$; the residual implicator corresponding to $\bodyfand{r}$ is denoted as $\rulefimp{r}$.
 A \bemph{constraint} is a rule with a constant in its head, whereas a \bemph{fact} is a rule with a constant as its body (and no
constant in its head). For convenience, we abbreviate a rule of the form $r: a \gets \pretnorm(b,1)$, with $b$ an extended literal, as $r: a \gets b$.
 
%
 
 A \bemph{FASP program} $P$ is a finite set of rules. We call a program \bemph{simple} if no rule contains negative extended literals. The set of all atoms occurring in $P$ is called the \bemph{Herbrand Base} $\hbase{P}$ of $P$. Note that the Herbrand Base is finite since we assume that no function symbols occur in the terms of the ungrounded program. For any $a \in \hbase{P}$ we define the set $P_a$ as the set of rules with atom $a$ in their head. An \bemph{interpretation} $I$ of $P$ is a fuzzy set in $\hbase{P}$. We extend interpretations to constants from $[0,1]$, extended literals, and rules as follows:
 \begin{enumerate}
   \item $I(c) = c$ if $c \in [0,1]$
   \item $I(\fneg{}{l}) = \;\fneg{}{I(l)}$ if $l$ is a literal
   \item $I(a \gets \bodyfand{}(b_1,\ldots,b_n)) = \rulefimp{r}(\bodyfand{}(I(b_1),\ldots,I(b_n)),I(a))$
  \end{enumerate}
 A \bemph{model} of a program $P$ is an interpretation $I$ of $P$ such that for each rule $r \in P$ we have $I(r) = 1$.
 
 Note that, although each rule in a FASP program can only have a single t-norm in its body, a rule with mixed t-norms, such as $r: a \gets \pretnorm_1(a,\pretnorm_2(b,c))$, can easily be simulated by introducing a polynomial number of new literals and rules. In the case of rule $r$ we need one new literal $a'$ and two new rules $r_1: a \gets \pretnorm_1(a,a')$ and $r_2: a' \gets \pretnorm_2(b,c)$.
 
 


\begin{example}\label{ex:prog1}
 Consider the program $P$, which consists of the following set of rules:
 \begin{align*}
  r_{1}: a &\gets \pretnorm_m(b,c)\\
  r_{2}: b &\gets 0.8\\
  r_{3}: c &\gets \pretnorm_m(a,\lneg{b})\\
  r_{4}: 0 &\gets \pretnorm_l(a,b)
 \end{align*}
 Note that rule $r_2$ is a fact, and rule $r_4$ a constraint.
The fuzzy sets $I_{1} = \{ a^{0}, b^{0.8}, c^{0} \}$, and $I_{2} = \{ a^{0.2}, b^{0.8}, c^{0.2} \}$ are interpretations of $P$.
For both interpretations we have that $I_1(r_1) = I_2(r_1) = \ldots = I_1(r_4) = I_2(r_4) = 1$, i.e.~they both are models of the program.
\end{example}

The definition of fuzzy answer sets relies on the notion of \bemph{unfounded sets}, studied in \cite{FASP:amai}, which correspond to sets of ``assumption atoms'' that have no proper motivation from the program.

\begin{definition}[Unfounded sets \cite{FASP:amai}]\label{def:unfounded}
 Let $P$ be a FASP program and let $I$ be an interpretation of $P$. A set $U \subseteq \hbase{P}$ is called \bemph{unfounded} w.r.t.~$I$ iff for each atom $u \in U$ and rule $r \in P_u$ either:
  \begin{enumerate}
   \item $U \cap \posbody{r} \neq \emptyset$; or
   \item $I(u) > I(\body{r})$; or
   \item $I(\body{r}) = 0$
  \end{enumerate}
 An interpretation $I$ of $P$ is called \bemph{unfounded-free} iff $\supp{I} \cap U = \emptyset$ for any set $U$ that is unfounded w.r.t.~$I$.
\end{definition}

Intuitively, an unfounded set w.r.t.~an interpretation $I$ of a FASP program $P$ is a set of atoms that obtain a value in $I$ that is not motivated by the rules of the program. The first condition of Definition~\ref{def:unfounded} ensures that the values of the literals in $U$ are justified by the values of literals not in $U$.
 The second condition shows that the degree to which a rule can motivate an atom is bounded by the value of its body. The third condition is needed to obtain a proper generalization of the classical definition of unfounded sets~\cite{BaralBook} (see~\cite{FASP:amai} for more details).


\begin{example}\label{ex:prog1-unfounded}
 Consider program $P$ and interpretations $I_1$ and $I_2$ from Example~\ref{ex:prog1}. For $I_{2}$ we can see that $U_{2} = \{a,c\}$ is an unfounded set, as for rule $r_{1}$ and $r_{3}$, the only rules with $a$ or $c$ in the head, we have that $\posbody{r_{1}} \cap U_{2} \neq \emptyset$ and $\posbody{r_{3}} \cap U_{2} \neq \emptyset$. Since $\supp{I_{2}} \cap U_{2} \neq \emptyset$, interpretation $I_{2}$ is not unfounded-free. Interpretation $I_{1}$, however, is unfounded-free.
\end{example}

%

As shown in \cite{FASP:amai}, \bemph{answer sets} of FASP programs can be defined as the unfounded-free interpretations, which reflects the intuition that each atom in an answer set should have a proper motivation from the program.

\begin{definition}[Answer Set \cite{FASP:amai}]\label{def:kanswersets}
 Let $P$ be a FASP program. A model $M$ of $P$ is called an \bemph{answer set} iff $M$ is unfounded-free.
\end{definition}

\begin{example}
Consider program $P$ and interpretation $I_1$ from Example~\ref{ex:prog1}. Since we know from Example~\ref{ex:prog1-unfounded} that $I_{1}$ is unfounded-free, it follows that $I_{1}$ is an answer set of $P$.
\end{example}

An alternative definition of answer sets, in terms of fixpoints, exists (see for example \cite{Luka06}). We will use this to generalize the ASSAT procedure described in \cite{assat-linzhao}. 

\begin{definition}[Immediate Consequence Operator \cite{Luka06}]\label{def:imcons}
 Let $P$ be a FASP program. The \bemph{immediate consequence operator} of $P$ is the mapping $\nfimcons{P} : \Fuzzy{\hbase{P}} \to \Fuzzy{\hbase{P}}$ defined by
  \begin{align*}
   \nfimcons{P}(I)(l) &= \sup \{ I(\body{r}) \mid r \in P_l \}
  \end{align*}
\end{definition}

As shown in \cite{Luka06}, for simple programs this operator is monotonic and thus has a least fixpoint \cite{tarski:lattice}, denoted as $\lfpnfimcons{P}$. For these simple programs, \cite{Luka06} then defines the answer sets of a program as the least fixpoints of this operator. Since this operator is monotonic, the least fixpoint is unique and can be found by iteratively applying $\nfimcons{P}$ from the interpretation $\emptyset$ until a fixpoint is encountered. For non-simple programs, \cite{Luka06} defines a reduct operation that transforms a non-simple program into a simple program.

\begin{definition}[Reduct \cite{Luka06}]\label{def:reduct}
 Let $P$ be a FASP program and let $r: a \gets \pretnorm(b_1,\ldots,b_m,\fneg{}{b_{m+1}},\ldots,\fneg{}{b_{n}})$ be a rule in $P$, where $(b_1,\ldots,b_m) = \posbody{r}$. The \bemph{reduct} of rule $r$, with respect to an interpretation $I$, is denoted as $r^I$, and defined by
   $$r^I: a \gets \pretnorm(b_1,\ldots,b_m,I(\fneg{}{b_{m+1}}),\ldots,I(\fneg{}{b_n}))$$
 The reduct of a program $P$ w.r.t.~an interpretation $I$ is defined as $P^I = \{ r^I \mid r \in P \}$.
\end{definition}

\begin{example}\label{ex:prog1-reduct}
Consider program $P$ and interpretation $I_1$ from Example~\ref{ex:prog1}. The reduct of $P$ with respect to $I_1$ then is the following program
 \begin{align*}
   r_{1}^{I_{1}}: a &\gets \pretnorm_{m}(b,c)\\
  r_{2}^{I_{1}}: b &\gets 0.8\\
  r_{3}^{I_{1}}: c &\gets \pretnorm_{m}(a,0.2)\\
  r_{4}^{I_{1}}: 0 &\gets \pretnorm_{l}(a,b)
 \end{align*}
\end{example}

In the following we show the novel result that the semantics in terms of fixpoints coincide with those in terms of unfounded sets.
An important lemma regarding the immediate consequence operator and reduct is the following.

\begin{lemma}\label{lem:notmodel-lfviolated}
 Let $P$ be a FASP program. For any interpretation $I$ of $P$ it holds that $I = \nfimcons{P}(I)$ iff $I = \nfimcons{P^I}(I)$.
\end{lemma}
\begin{proof}
 Follows trivially by the construction of $P^I$ and Definition~\ref{def:reduct}.
\end{proof}

We now show that any answer set is a fixpoint of the immediate consequence operator.

\begin{lemma}\label{lem:unfounded-fixpoint}
 Let $P$ be a FASP program. Then any answer set $A$ of $P$ is a fixpoint of $\nfimcons{P}$.
\end{lemma}
\begin{proof}
 Let $A$ be an answer set of $P$. We show that $A(a) = \sup \{ A(\body{r}) \mid r \in P_a \} = \nfimcons{P}(A)(a)$ for any $a \in \hbase{P}$, from which the stated readily follows. The proof is split into the case for $a \in \supp{A}$ and $a \not\in \supp{A}$. For any $a \in \supp{A}$ it must hold that $\{a\}$ is not unfounded w.r.t.~$A$,
 meaning that $P_a \neq \emptyset$ and there is some $r \in P_a$ such that $A(a) \leq A(\body{r})$. Since $A(r) = 1$, it then follows from (\ref{eq:resfimpprop}) that $A(a) = A(\body{r})$. As for any $r' \in P_a$ we have $A(r') = 1$, from (\ref{eq:resfimpprop}) it also follows that $A(\body{r}) = A(a) \geq A(\body{r'})$. Hence $A(\body{r}) = A(a)$ is the supremum of $\{ A(\body{r'}) \mid r' \in P_a \}$.

 The case for $a \not\in \supp{A}$ is as follows. First remark that as $A(a) = 0$, it follows from (\ref{eq:resfimpprop}) and the fact that $A(r) = 1$ for each $r \in P_a$, that $A(\body{r}) = 0$. Hence, $A(a) = \sup \{ A(\body{r}) \mid r \in P_a \}$. 
\end{proof}

Second we show that answer sets can be characterized in terms of fixpoints of the immediate consequence operator.


\begin{proposition}\label{prop:fixp-charact}
 Let $P$ be a FASP program. An interpretation $A$ is an answer set of $P$ iff $A = \lfpnfimcons{P^A}$.
\end{proposition}
\begin{proof}
 Let $M$ be a model of $P$.
 In \cite{FASP:amai} it was shown that the least fixpoint of $\nfimcons{P^M}$ must necessarily be unfounded-free (Proposition~4). As any fixpoint of $\nfimcons{P}$ is a model of $P$, we only need to show that if $M$ is unfounded-free, it is the least fixpoint of $\nfimcons{P^M}$.
 
 Suppose $M \neq \lfpnfimcons{P^M}$. Then, since any unfounded-free model is a fixpoint of $\nfimcons{P^M}$ due to Lemmas~\ref{lem:notmodel-lfviolated} and ~\ref{lem:unfounded-fixpoint}, it holds that some set $M' \subset M$ exists such that $M' = \lfpnfimcons{P^M}$.
 
 Consider then $U = \{ u \in \hbase{P} \mid M'(u) < M(u) \}$. Surely $U \subseteq \supp{M}$ and hence $U \cap \supp{M} \neq \emptyset$. We now show that $U$ is unfounded with respect to $M$, leading to a contradiction. First, we show that for any atom $u \in U$ and rule $r \in P_u$ it holds that 
  \begin{equation}\Big(\posbody{r} \cap U = \emptyset\Big) \imp M(\body{r}) < M(u)\label{eq:prop:fixp-charact}\end{equation}
 as follows
  \[\begin{array}{rcll}
   & & &\posbody{r} \cap U = \emptyset\\
    &\equiv& \hint{Def.~$\cap$}& \not\Exists{l \in \posbody{r}}{l \in U}\\
    &\equiv& \hint{Duality $\forall$,$\exists$}& \Forall{l \in \posbody{r}}{l \not\in U}\\
    &\equiv& \hint{Def.~$U$}& \Forall{l \in \posbody{r}}{M(a) = M'(a)}\\
    &\imp& \hint{$M(l) = M(l^M)$}& M(\body{r^M}) = M'(\body{r^M})\\
    &\imp& \hint{Monotonicity $\sup$}& M(\body{r^M}) \leq \sup_{r' \in P^M_u} M'(\body{r'})\\
    &\equiv& \hint{$M' = \nfimcons{P^M}(M')$}& M(\body{r^M}) \leq M'(u)\\
    &\imp& \hint{$u \in U$, Def.~$U$}& M(\body{r^M}) < M(u)\\
  \end{array}\]
\noindent Thus, since it follows from the Definition of $r^M$ that $M(\body{r^M}) = M(\body{r})$, we have shown that (\ref{eq:prop:fixp-charact}) holds. From this equation we obtain that
  $$\Big(\posbody{r} \cap U \neq \emptyset\Big) \vee \Big(M(\body{r}) < M(u)\Big)$$
 Hence
  $$\Big(\posbody{r} \cap U \neq \emptyset\Big) \vee \Big(M(\body{r}) < M(u)\Big) \vee \Big(M(\body{r}) = 0\Big)$$
 Which means $U$ is unfounded with respect to $M$, a contradiction.
\end{proof}


\begin{example}
Consider program $P$ and interpretations $I_1$ and $I_2$ from Example~\ref{ex:prog1}. Computing the least fixpoint of $\nfimcons{P^{I_1}}$ and $\nfimcons{P^{I_2}}$ can be done by iteratively applying these operators, starting from $\emptyset$, until we find a fixpoint. Hence for $P^{I_{1}}$ we obtain in the first iteration $J_{1} = \nfimcons{P^{I_{1}}}(\emptyset) = \{ b^{0.8} \}$. The second iteration gives $J_{2} = \nfimcons{P^{I_{1}}}(J_{1}) = \{ b^{0.8} \} = J_{1}$, hence a fixpoint, meaning $\{ b^{0.8} \} = I_{1}$ is the least fixpoint of $\nfimcons{P^{I_{1}}}$.
Iteratively applying $\nfimcons{P^{I_{2}}}$ brings us $J_{1} = \nfimcons{P^{I_{2}}}(\emptyset) = \{ b^{0.8} \}$, which is also a fixpoint of $\nfimcons{P^{I_{2}}}$.
\end{example}

\section{Completion of FASP programs}\label{sec:completion}

In this section we show how certain fuzzy answer set programs can be translated to fuzzy theories such that the models of these theories correspond to answer sets of the program and vice versa. Such a correspondence is important as it allows us to find answer sets using fuzzy SAT solvers.

\begin{definition}[Completion of a FASP program]\label{def:kcompletion}
 Let $P$ be a FASP program. The completion of $P$, denoted as $\comp{P}$, is defined as the following set of fuzzy formulas:
  $$\{ a \feq (\max \{ \body{r} \mid r \in P_a \}) \mid a \in \hbase{P} \} \cup \{ \prefimp_r(\body{r},\head{r}) \mid r \in P, \head{r} \in [0,1] \}$$
 where $\feq$ is the biresiduum of an arbitrary residual implicator, and $\prefimp_r$ is the residual implicator of the t-norm used in the body of rule $r$.
\end{definition}

The completion of a program consists of two parts, viz.~a part for the literals $\{ a \feq (\max \{ \body{r} \mid r \in P_a \}) \mid a \in \hbase{P} \}$, and a part for constraints $\{ \prefimp_r(\body{r},\head{r}) \mid r \in P, \head{r} \in [0,1] \}$. The constraints part simply ensures that all constraints are satisfied. The literal part ensures two things. By definition of the biresiduum and the fact that $\prefimp(a,b) = 1$ iff $I(a) \leq I(b)$ for any residual implicator, we have that $I(a \feq b) = 1$ iff $I(a) \leq I(b)$ and $I(b) \leq I(a)$. Hence, the literal part of the completion establishes that rules are satisfied and second that the value of the literal is not higher than what is supported by the rule bodies.

\begin{example}\label{ex:prog1-completion}
 Consider program $P$ from Example~\ref{ex:prog1}. Its completion is the following set of fuzzy propositions
 \begin{align*}
  &a \feq \pretnorm_m(b,c)\\
  &b \feq 0.8)\\
  &c \feq \pretnorm_m(a,\mneg{b})\\
  & \prefimp_l(\pretnorm_l(a,b),0)
 \end{align*}
\end{example}

Note that when applying Definition~\ref{def:kcompletion} for a literal $l$ that does not appear in the head of any rule, we get $a \feq \max \emptyset$, where we define $\max \emptyset = 0$.

We can now show that any answer set of a program $P$ is a model of its completion $\comp{P}$.


\begin{proposition}\label{prop:kansset-is-modelcomp}
 Let $P$ be a FASP program and let $\comp{P}$ be its completion. Then any answer set of $P$
 is a model of $\comp{P}$.
\end{proposition}
\begin{proof}
 Suppose $A$ is an answer set of $P$. By Lemma~\ref{lem:unfounded-fixpoint}, it follows that $A$ is a fixpoint of $\nfimcons{P}$, hence for each $a \in \hbase{P}$, $A(a) = \sup \{ A(\body{r}) \mid r \in P_a \}$. By construction of $\comp{P}$ and the fact that $A$ is a model of $P$, it then easily follows that $A \models \comp{P}$.
%
%
\end{proof}

\begin{example}
Consider program $P$ and interpretation $I_{1}$ from Example~\ref{ex:prog1}. It is easy to see that $I_{1}$ is a model of $\comp{P}$.
\end{example}

The reverse of Proposition~\ref{prop:kansset-is-modelcomp} is not true in general, which is unsurprising because it is already invalid for classical answer set programming. The problem occurs for programs with ``loops'', as shown in the following example.

\begin{example}
 Consider program $P$ and interpretation $I_{2}$ from Example~\ref{ex:prog1}. We can easily see that $I_{2}$ is a model of $\comp{P}$, but, as we have seen in Example~\ref{ex:prog1-unfounded}, it is not an answer set of $P$.
\end{example}

One might wonder whether taking the minimal models of the completion would solve the above problem. The following example shows that the answer is negative.

\begin{example}\label{ex:progmin}
 Consider the following program $P_{\mathit{min}}$
 \begin{align*}
  a &\gets a\\
  p &\gets \pretnorm_l(\lneg{p},\lneg{a})
 \end{align*}
 \noindent The completion $\comp{P_{\mathit{min}}}$ is
 \begin{align*}
  a &\feq a\\
  p &\feq \pretnorm_l(\lneg{p},\lneg{a})
 \end{align*}
 \noindent Consider now the interpretation $I = \{a^{0.2},p^{0.4}\}$. Since $I(a) = I(a)$ and 
  $$\pretnorm_l(\lneg{I(p)},\lneg{I(a)}) = \max(0,1-I(p) + 1-I(a) -1) = 0.4$$
 we can see that $I$ is a model of $\comp{P_{\mathit{min}}}$. We show that it is a minimal model as follows. Suppose some $I' \subset I$ exists. Then we can consider three cases: (i) $I'(a) < I(a)$ and $I'(p) = I(p)$; (ii) $I'(a) = I(a)$ and $I'(p) < I(p)$; (iii) $I'(a) < I(a)$ and $I'(p) < I(p)$. In all three cases we obtain that $\pretnorm_l(\lneg{I'(p)},\lneg{I'(a)}) > 0.4 > I'(p)$, since $\lneg{I'(a)} = 1-I'(a) > 0.8$ or $\lneg{I'(p)} > 0.6$. Hence $I'$ is not a model of $\comp{P_{\mathit{min}}}$ and $I$ is thus a minimal model of $\comp{P_{\mathit{min}}}$.
 However, $I'$ is not an answer set of $P_{\mathit{min}}$ since $\lfpimcons{P_{\mathit{min}}^I} = \{a^0,p^{0.4}\}$.
\end{example}

%
%

As in the crisp case however, when a program has no loops in its positive dependency graph, the models of the completion and the answer sets coincide. First we define exactly what a loop of a FASP program is, and then we show that this property indeed still holds for FASP.

\begin{definition}[Loop]\label{def:loop}
 Let $P$ be a FASP program. The \bemph{positive dependency graph} of $P$ is a directed graph $\depgraph{P} = \tuple{\hbase{P},D}$ where $(a,b) \in D$ iff $\Exists{r \in P_a}{b \in \poslit{\body{r}}}$. For ease of notation we also denote this relation with $(a,b) \in \depgraph{P}$ for atoms $a$ and $b$ in the Herbrand base of $P$. We call a non-empty set $L \subseteq\hbase{P}$ a \bemph{loop} of $P$ iff for all literals $a$ and $b$ in $L$ there is a path (with length $> 0$) from $a$ to $b$ in $\depgraph{P}$ such that all vertices of this path are elements of $L$.
\end{definition}

\begin{example}
 Consider program $P_{min}$ from Example~\ref{ex:progmin}. The dependency graph of $P_{min}$ is pictured in Figure~\ref{fig:depgraph-prog2}. We can see that $\{a\}$ is a loop. If this loop was not in the program, its completion would become
 \begin{align*}
  a &\feq 0\\
  p &\feq \pretnorm_m(\fneg{m}{p},\fneg{m}{a})
 \end{align*}
 This fuzzy theory has no models. Since program $P_{min}$ has no answer sets, this means the answer sets coincide with the models of the completion when removing the loop.
\end{example}

\begin{example}
Consider program $P$ from Example~\ref{ex:prog1}. The dependency graph of $P$ is pictured in Figure~\ref{fig:depgraph-prog1}. We can clearly see that there is a loop between nodes $a$ and $c$. Due to this loop, the values of $a$ and $c$ are not sufficiently constrained in the completion.
\end{example}

From the preceding examples one might think that removing the loops from the program would be sufficient to make the models of the completion and the answer sets coincide. However, this is not the case, as the semantics of the program then changes, as illustrated in the following example.

\begin{example}\label{ex:progchange}
 Consider program $P_{change}$ consisting of the following rules
 \begin{align*}
  r_1: a &\gets 0.3\\
  r_2: a &\gets b\\
  r_3: b &\gets a
 \end{align*}
 Its single answer set is $\{ a^{0.3}, b^{0.3} \}$. If we remove rule $r_2$ or $r_3$, the answer set of the resulting program is $\{ a^{0.3} \}$.
\end{example}

\begin{figure}
 \centering
 \includegraphics[scale=0.50]{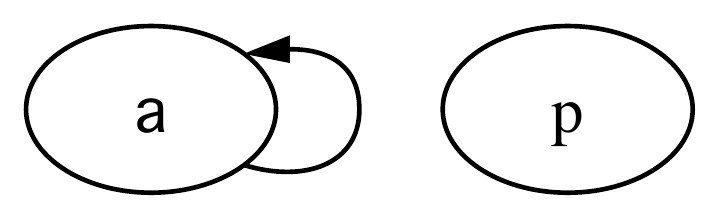}
 \caption{Dependency graph of program $P_{min}$ from Example~\ref{ex:progmin}}
 \label{fig:depgraph-prog2}
\end{figure}

\begin{figure}
 \centering
 \includegraphics[scale=0.50]{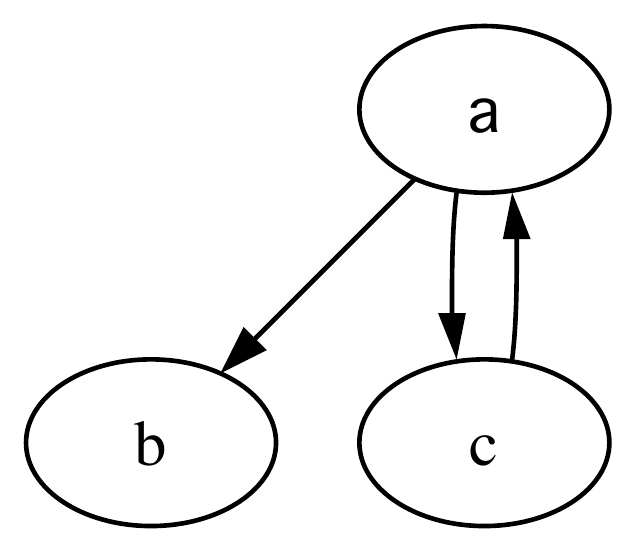}
 \caption{Dependency graph of program~$P$ from Example~\ref{ex:prog1}}
 \label{fig:depgraph-prog1}
\end{figure}

We can now show that for programs without loops the answer sets coincide with the models of their completion. We first introduce two lemmas.

\begin{lemma}\label{lem:noloop-modelcomp-is-ansset-1}
 Let $G = \tuple{V,E}$ be a directed graph with a finite set of vertices and $X \subseteq V$ with $X \neq \emptyset$. If every node in $X$ has at least one outgoing edge to another node in $X$, there must be a loop in $X$.
\end{lemma}
\begin{proof}
 From the assumptions it holds that each $x \in X$ has an outgoing edge to another node in $X$. This means that there is an infinite sequence of nodes $x_1,x_2,\ldots$ such that $(x_i,x_{i+1}) \in E$ for $i \geq 1$. Since $X$ is finite, it follows that some vertex occurs twice in this sequence, and hence that there is a loop in $X$.
\end{proof}

\begin{lemma}\label{lem:noloop-modelcomp-is-ansset-2}
 Let $P$ be a FASP program, $I$ an interpretation of $P$ and $U \subseteq \hbase{P}$. Then if $I \models \comp{P}$ and $U$ is unfounded w.r.t.~$I$ it holds that for each $u$ in $U \cap \supp{I}$ there is some $r$ in $P_u$ such that $\posbody{r} \cap U \cap \supp{I} \neq \emptyset$. 
\end{lemma}
\begin{proof}
 Assume that $u \in U \cap \supp{I}$, in other words $u \in U$ and $I(u) > 0$. As $U$ is unfounded w.r.t.~$I$, for each $r \in P_u$ it holds that either\\
   \phantom{xx}1. $\posbody{r} \cap U \neq \emptyset$; or\\
   \phantom{xx}2. $I(\body{r}) < I(u)$; or\\
   \phantom{xx}3. $I(\body{r}) = 0$\\
  We can now show that there is at least one rule $r \in P_u$ that violates the second and third of these conditions, meaning it must satisfy the first. 
  
  From $I \models \comp{P}$ we know by construction of $\comp{P}$ that for each $u \in U$, $I(u) = \sup \{ I(\body{r}) \mid r \in P_u\}$. Hence for each $u \in U$ there is a rule $r \in P_u$ such that $I(u) = I(\body{r})$, thus the second condition is violated. Since $I(u) > 0$, it then also follows that the third condition is violated.
  
  In other words there must be some $r \in P_u$ such that $I(\body{r}) = I(u)$ and $I(\body{r}) \neq 0$. Since $U$ is unfounded w.r.t.~$I$, this means that $\posbody{r} \cap U \neq \emptyset$. Since $I(\body{r}) \neq 0$ implies that $\posbody{r} \subseteq \supp{I}$ due to the fact that $\pretnorm(0,x) = 0$ for any t-norm $\pretnorm$, we can conclude that there is some $r \in P_u$ such that $\posbody{r} \cap U \cap \supp{I} \neq \emptyset$.
\end{proof}

Using these lemmas we can now show that the answer sets of any program without loops in its dependency graph coincide with the models of its completion. This resembles Fages' theorem on tight programs in classical ASP \cite{fages:completion}.

\begin{proposition}\label{prop:noloops-ansset-is-modelcomp}
 Let $P$ be a FASP program. If $P$ has no loops in its positive dependency graph it holds that an interpretation $I$ of $P$ is
 an answer set of $P$ iff $I \models \comp{P}$.
\end{proposition}
\begin{proof}
 We already know from Proposition~\ref{prop:kansset-is-modelcomp} that any answer set of $P$ is necessarily a model of $\comp{P}$, hence we only need to show that every model of $\comp{P}$ is an answer set of $P$ under the conditions of this proposition. As $I \models \comp{P}$, it holds that $I$ is a model of $P$. We show by contradiction that $I$ is unfounded-free. Assume that there is a set $U \subseteq \hbase{P}$ such that $U$ is unfounded w.r.t.~$I$ and $U \cap \supp{I} \neq \emptyset$. From Lemma~\ref{lem:noloop-modelcomp-is-ansset-2} we know that for each $u \in U \cap \supp{I}$ it holds that there is some rule $r \in P_u$ such that $\posbody{r} \cap U \cap \supp{I} \neq \emptyset$. Using the definition of $\depgraph{P}$, this means that for each such $u$ there is some $u' \in U \cap \supp{I}$ such that $\depgraph{P}(u,u')$. This however means that there is a loop in $\depgraph{P}$ by Lemma~\ref{lem:noloop-modelcomp-is-ansset-1}, contradicting the assumption.
\end{proof}

Hence, finding the answer sets of a program with no loops in its positive dependency graph can be done by finding models of its completion.

\section{Loop Formulas}\label{sec:loopelimination}

As mentioned in the previous section, sometimes the models of the completion are not answer sets. In this section, we investigate how the solution that has been proposed for boolean answer set programming, viz.~adding loop formulas to the completion \cite{assat-linzhao}, can be extended to fuzzy answer set programming.

For this extension, we start from a partition of the rules whose heads are in some particular loop $L$. Based upon this partition, for every loop $L$ we define a formula in fuzzy logic, such that any model of the completion satisfying these formulas is an answer set.

For any program $P$ and loop $L$ we consider the following partition of the rules in $P$ whose head belongs to the set $L$ (due to \cite{assat-linzhao})
 \begin{eqnarray}
  \looprules{P}{L} &= \{ a \gets B \mid ((a \gets B) \in P) \wedge (a \in L) \wedge (B^+ \cap L \neq \emptyset) \}\\
  \nonlooprules{P}{L} &= \{ a \gets B \mid ((a \gets B) \in P) \wedge (a \in L) \wedge (B^+ \cap L = \emptyset) \}
 \end{eqnarray}

Note that this partition only takes the positive occurrences of atoms in the loop into account.
Intuitively, the set $\looprules{P}{L}$ contains the rules that are ``in'' the loop $L$, i.e.~the rules that are jointly responsible for the creation of the loop in the positive dependency graph, whereas the rules in $\nonlooprules{P}{L}$ are the rules that are outside of this loop. We will refer to them as ``loop rules'', resp.~``non-loop rules.''

\begin{example}
Consider program $P$ from Example~\ref{ex:prog1}. It is clear that for the loop $L = \{a,c\}$ the set of loop rules is $\looprules{P}{L} = \{ r_{1},r_{3} \}$ and the set of non-loop rules is $\nonlooprules{P}{L} = \emptyset$.
\end{example}


\begin{example}
Consider program $P$ from Example~\ref{ex:prog1} with interpretations $I_{1}$ and $I_{2}$ from Example~\ref{ex:prog1} once again. 
It is clear that in $I_1$ no loop rules were used to derive the values of $a$ and $c$, whereas in $I_2$ only loop rules are used. 
\end{example}


Hence there is a problem when the value of literals in a loop are only derived from rules in the loop. To solve this problem, we should require that at least one non-loop rule motivates the value of these loop literals. As illustrated in the next example, one non-loop rule is sufficient as the value provided by this rule can propagate through the loop by applying loop rules.

\begin{example}\label{ex:prog2}
 Consider program $P_{change}$ from Example~\ref{ex:progchange} again.
 Clearly this program has a loop $L = \{ a, b \}$ with $\looprules{P}{L} = \{ r_{2}, r_{3} \}$ and $\nonlooprules{P}{L} = \{ r_{1} \}$. Consider then interpretations $I_{1} = \{ a^{0.3}, b^{0.3} \}$ and $I_{2} = \{ a^{1}, b^{1} \}$. We can easily see that $I_{1}$ is an answer set of $P$, whereas $I_{2}$ is not, although they are both models of $\comp{P}$. The problem is that in $I_{2}$ the values of $a$ and $b$ are higher than what can be derived from the non-loop rule $r_1$, whereas in $I_1$ their values are exactly what can be justified from applying rule $r_1$.
 The latter is allowed, as values are properly supported from outside the loop, while the former is not, as in this case the loop is ``self-motivating''.
\end{example}

To remove the non-answer set models of the completion, we add loop formulas to the completion, defined as follows.

\begin{definition}[Loop Formula]\label{def:loopformula}
 Let $P$ be a FASP program and $L = \{ l_1,\ldots,l_m \}$ a loop of $P$. Suppose that $\nonlooprules{P}{L} = \{ r_1,\ldots,r_n \}$. Then the loop formula induced by loop $L$, denoted by $\loopform{L}{P}$, is the following fuzzy logic formula:
  \begin{equation}\prefimp(\max(l_1,\ldots,l_m),\max(\body{r_1},\ldots,\body{r_n})\label{eq:loopforms}\end{equation}
 where $\prefimp$ is an arbitrary residual implicator. If $\nonlooprules{P}{L} = \emptyset$, the loop formula becomes
  $$\prefimp(\max(l_1,\ldots,l_m),0)$$
\end{definition}

 The loop formula proposed for boolean answer set programs in \cite{assat-linzhao} is of the form
  \begin{equation}\neg(\bigwedge \body{r_{1}} \vee \ldots \vee \bigwedge \body{r_{n}} ) \imp (\neg l_1 \wedge \ldots \wedge \neg l_m)\label{eq:loopforms-2}\end{equation}
  
It can easily be seen that (\ref{eq:loopforms}) is a straightforward generalisation of (\ref{eq:loopforms-2}) as the latter is equivalent to
 $$(l_1 \vee \ldots \vee l_m) \imp (\bigwedge \body{r_{1}} \vee \ldots \vee \bigwedge \body{r_{n}})$$
Note that this equivalence is preserved in \L ukasiewicz logic, but not in G\"odel or product logic.
 Furthermore, since $I \models \prefimp(\max(l_1,\ldots,l_m),0)$ only when $\max(I(l_1),$ $\ldots,$ $I(l_m))$ $\leq 0$, it is easy to see that in the case where $\nonlooprules{P}{L} = \emptyset$, the truth value of all atoms in the loop $L$ is $0$.
 
\begin{example}
 Consider program $P$ and interpretations $I_1$ and $I_2$ from Example~\ref{ex:prog1}. The loop formula for its loop $L = \{ a,c \}$ is the fuzzy formula $\prefimp_m(\max(a,c),0)$, since $\nonlooprules{P}{L} = \emptyset$. It is easy to see that $I_{2}$ does not satisfy this formula, while interpretation $I_{1}$ does. 
\end{example}

\begin{example}
 Consider program $P_{change}$ from Example~\ref{ex:progchange}. The loop formula for its loop $L = \{ a,b \}$ is the propositional formula $\prefimp_m(\max(a,b),0.3)$, since $\nonlooprules{P}{L} = \{ r_{1} \}$. Again we see that interpretation $I_{1}$ from Example~\ref{ex:prog2} satisfies this loop formula, whereas interpretation $I_{2}$ from the same example does not.
\end{example}
 
 We now show that by adding loop formulas to the completion of a program, we get a fuzzy propositional theory that is both sound and complete with respect to the answer set semantics. First we show that this procedure is complete.
 

\begin{proposition}[Completeness]\label{prop:loopforms-complete}
 Let $P$ be a FASP program, let $\mathcal{L}$ be the set of all loops of $P$, and define $\loopformprog{P} = \{ \loopform{L}{P} \mid L \in \mathcal{L} \}$.  For any answer set $I$ of $P$, it holds that $I \models \loopformprog{P} \cup \comp{P}$. 
\end{proposition}
\begin{proof}
 Suppose $I$ is an answer set of $P$ and $I \not\models \loopformprog{P} \cup \comp{P}$. Since any answer set is a model of $\comp{P}$ according to Proposition~\ref{prop:kansset-is-modelcomp}, this means that $I \not\models \loopformprog{P}$. Hence, the loop formula of some loop $L$ in $P$ is not fulfilled; this means:
  $$\sup_{u \in L} I(u) > \sup_{r \in \nonlooprules{P}{L}} I(\body{r})$$
 Consider then the set $U = \{ u \in L \mid I(u) > \sup_{r \in \nonlooprules{P}{L}} I(\body{r}) \}$.
 We show that $U$ is unfounded w.r.t.~$I$, i.e.~we show that for each $u \in U$ and rule $r \in P_u$, at least one of the conditions of Definition~\ref{def:unfounded} applies.
 
 Since $P_u = \looprules{P_u}{L} \cup \nonlooprules{P_u}{L}$, each rule $r \in P_u$ must either be in $\looprules{P_u}{L}$ or in $\nonlooprules{P_u}{L}$. We consider the following cases:\\
   \phantom{xx}1.~If $r \in \nonlooprules{P_u}{L}$ then by construction of $U$ it holds that $I(\body{r}) < I(u)$.\\
  \phantom{xx}2.~If $r \in \looprules{P_u}{L}$ and $I(\body{r}) \leq \sup_{r' \in \nonlooprules{P_u}{L}} I(\body{r'})$, by construction of $U$ we have that $I(\body{r}) < I(u)$.\\
   \phantom{xx}3.~Suppose $r \in \looprules{P_u}{L}$ and $I(\body{r}) > \sup_{r' \in \nonlooprules{P_u}{L}} I(\body{r'})$. Since $\pretnorm(x,y) \leq \min(x,y)$ for each t-norm $\pretnorm$, we know that $I(\body{r}) \leq I(l)$ for each $l \in \posbody{r}$. Hence for each $l \in \posbody{r}$ we have $I(l) > \sup_{r' \in \nonlooprules{P_u}{L}} I(\body{r'})$. This means that, since $r \in \looprules{P}{L}$ and thus $\posbody{r} \cap L \neq \emptyset$, we know from the definition of $U$ that $\posbody{r} \cap U \neq \emptyset$.

Now remark that $U \cap \supp{I} \neq \emptyset$ as $U \subseteq \supp{I}$ due to $I(u) > 0$ for each $u \in U$.
From the above we can thus conclude that $U$ is unfounded w.r.t.~$I$, and since $U \cap \supp{I} \neq \emptyset$, that $I$ is not unfounded-free: a contradiction.
\end{proof}

Second we show that adding the loop formulas to the completion of a program is a sound procedure.

\begin{lemma}\label{lem:loopforms-sound-1}
 Let $G = \tuple{V,E}$ be a directed graph and $X \subseteq V$, with $V$ finite, such that each node of $X$ has at least one outgoing edge to another node in $X$. Then there is a set $L \subseteq X$ such that $L$ is a maximal loop in $X$ and for each $l \in L$ we have that there is no $x \in X \setminus L$ for which $(l,x) \in E$.
\end{lemma}
\begin{proof}
 From Lemma~\ref{lem:noloop-modelcomp-is-ansset-1} we already know that there must be a loop in $X$. Hence, there must also be a maximal loop in $X$. First, remark that maximal loops must of course be disjoint as otherwise their union would form a bigger loop.
 Consider then the set $X$, which is a collection of disjoint maximal loops $L$ and remaining nodes $S$ (single nodes that are not in any loop). There is an induced graph $G'$ of $G$ with nodes $S \cup L$ (i.e.~each maximal loop is a single node in the induced graph) and edges $E$ induced as usual (i.e.~$(L_1,L_2) \in E$ if for some node $l_1$ in $L_1$ there is a node $l_2$ in $L_2$ such that $(l_1,l_2) \in E$ and likewise for the nodes in $S$). Clearly, $G'$ is acyclic as otherwise the nodes in $G'$ on the cycle would create a bigger loop in $X$. Hence, $G'$ has leafs without outgoing edges. However, a leaf cannot be in $S$ since that would imply a node in $X$ without an outgoing edge. Thus we can conclude that all leafs in $G'$ are maximal loops in $X$.
\end{proof}

\begin{proposition}[Soundness]\label{prop:loopforms-sound}
 Let $P$ be a FASP program and let $\loopformprog{P}$ be the set of all loop formulas of $P$. Then for any interpretation $I$ of $P$ it holds that if $I \models \loopformprog{P} \cup \comp{P}$, then $I$ must be an answer set of $P$.
\end{proposition}

\begin{proof}
 Suppose $I \models \loopformprog{P} \cup \comp{P}$ and $I$ is not an answer set of $P$. Since any model of $\comp{P}$ must be a model of $P$, this must mean that $I$ is not unfounded-free, i.e.~that there exists a set $U \subseteq \hbase{P}$ such that $U$ is unfounded w.r.t.~$I$. From Lemma~\ref{lem:noloop-modelcomp-is-ansset-2} we know that for each $u \in U \cap \supp{I}$ there must be some $r \in P_u$ such that $\posbody{r} \cap U \cap \supp{I} \neq \emptyset$. Hence, by definition of $G_P$ this means that for each $u \in U \cap \supp{I}$ there is some $u' \in U \cap \supp{I}$ such that $(u,u') \in G_P$. Using Lemma~\ref{lem:loopforms-sound-1} this means that there is a set $L \subseteq U \cap \supp{I}$ such that $L$ is a loop in $P$ and for each $l \in L$ there is no $u \in (U \cap \supp{I})\setminus L$ such that $(l,u) \in G_P$. In other words, for each $l \in L$ and rule $r \in P_l$ we have that
  \begin{equation}\Big(U \cap \supp{I} \cap \posbody{r} \neq \emptyset\Big) \imp \Big(L \cap \posbody{r} \neq \emptyset\Big)\label{eq:prop:loopforms-sound-1}\end{equation}

 Now, consider $l \in L$. Since $L \subseteq U \cap \supp{I}$, we know that $I(l) > 0$. Hence, if $I(\body{r}) = I(l)$ for some rule $r \in P_l$, we know that $I(\body{r}) > 0$. As $U$ is unfounded w.r.t.~$I$, it follows from Definition~\ref{def:unfounded} that $L \cap \posbody{r} \neq \emptyset$.

 Using contraposition, this means that for each $l \in L$ and $r \in P_l$ we have that
  \begin{equation}\Big(L \cap \posbody{r} = \emptyset\Big) \imp \Big(I(\body{r}) \neq I(l)\Big)\label{eq:prop:loopforms-sound-2}\end{equation}

 By the definition of $\comp{P}$, however, we know that for each model of $\comp{P}$ and for each $a \in \hbase{P}$ and $r \in P_a$ we have $I(a) \geq I(\body{r})$. Hence for each $l \in L$ and $r \in P_l$ from (\ref{eq:prop:loopforms-sound-2}) we have that
  \begin{equation}\Big(L \cap \posbody{r} = \emptyset\Big) \imp \Big(I(\body{r}) < I(l)\Big)\label{eq:prop:loopforms-sound-3}\end{equation}
 Now, for each $l \in L$ and $r \in \nonlooprules{P}{L} \cap P_l$ by definition of $\nonlooprules{P}{L}$ it holds that $L \cap \posbody{r} = \emptyset$, meaning $I(\body{r}) < I(l)$. Thus, $\sup \{ I(\body{r}) \mid r \in \nonlooprules{P}{L} \} < \sup \{ I(l) \mid l \in L \}$, meaning $I \not\models \loopform{L}{P}$, a contradiction.
\end{proof}

A straightforward procedure for finding answer sets would now be to extend the completion of a program with all possible loop formulas and let a fuzzy SAT solver generate models of the resulting fuzzy propositional theory. The models of this theory are the answer sets of the program, as ensured by Propositions~\ref{prop:loopforms-complete} and \ref{prop:loopforms-sound}. As there may be an exponential number of loops, however, this translation is not polynomial in general. A similar situation arises for classical ASP. The solution proposed in \cite{assat-linzhao} overcomes this limitation by iteratively adding loop formulas. In particular, a SAT solver is first used to find a model of the completion of a classical ASP program. Then it is checked in polynomial time whether this model is an answer set. If this is not the case, a loop formula, which is not satisfied by the model that was found, is added to the completion. The whole process is then repeated until an answer set is found. We will show that a similar procedure can be used to find answer sets of a FASP program.  

Starting from the fixpoint characterization of answer sets of FASP programs, we show that for any given model of the completion that is not an answer set, we can construct a loop that is violated.

\begin{proposition}\label{prop:notmodel-lfviolated}
 Let $P$ be a FASP program. If an interpretation $I$ of $P$ is a model of $\comp{P}$ and $I \neq \lfpnfimcons{P^I}$, then some $L \subseteq \supp{I \fsetminus \lfpnfimcons{P^I}}$ must exist such that $I \not\models \loopform{P}{L}$.
\end{proposition}
\begin{proof}
 Suppose $I$ is an interpretation of $P$ and $I \models \comp{P}$, then from the definition of $\comp{P}$ and Lemma~\ref{lem:notmodel-lfviolated}, we can easily see that $I$ is a fixpoint of $\nfimcons{P^I}$. Since $I \neq \lfpnfimcons{P^I}$, some $I' \subset I$ must exist such that $I' = \lfpnfimcons{P^I}$.
 
 Consider then the set $U = \{ u \in \hbase{P} \mid I(u) > I'(u) \}$. It holds that $U = \supp{I \fsetminus I'}$ since $I' \subset I$ and thus $U = \supp{I \fsetminus \lfpnfimcons{P^I}}$ by definition of $I'$. From the proof of Proposition~\ref{prop:fixp-charact} we then also know that for this set $U$ the following property holds
  \begin{equation}\Forall{u \in U}{\Forall{r \in P_u}{\Big(\posbody{r} \cap U = \emptyset\Big) \imp \Big(I(\body{r}) < I(u)\Big)}}\label{eq:notmodel-lfviolated-1}\end{equation}
 We can then show that there is a loop in $U$ whose loop formula is violated. Since $I = \nfimcons{P^{I}}(I)$ we know from Lemma~\ref{lem:notmodel-lfviolated} that $I = \nfimcons{P}(I)$. From the definition of $\nfimcons{P}$ this means
  $$\Forall{l \in \hbase{P}}{I(l) = \sup \{I(\body{r}) \mid r \in P_l \}}$$
 Since the supremum is attained because $P$ is finite we obtain
  $$\Forall{l \in \hbase{P}}{\Exists{r \in P_l}{I(l) = I(\body{r})}}$$
 As $U \subseteq \hbase{P}$ this means
  $$ \Forall{u \in U}{\Exists{r \in P_u}{I(l) = I(\body{r})}}$$
 Using (\ref{eq:notmodel-lfviolated-1}) it then holds that
   $$\Forall{u \in U}{\Exists{r \in P_u}{\posbody{r} \cap U \neq \emptyset}}$$
  From the definition of $G_{P}$ we thus get
   $$\Forall{u \in U}{\Exists{u' \in U}{(u,u') \in G_{P}}}$$
  Using Lemma~\ref{lem:loopforms-sound-1} it follows that there is a set $L \subseteq U$ that is a loop in $P$ such that for each $l \in L$ there is no $l' \in U\setminus L$ such that $(l,l') \in E$. In other words, for each $l \in L$ there is no $l' \in U\setminus L$ such that there is a rule $r \in P_l$ for which $l' \in \posbody{r}$. Hence for each $l \in L$ and rule $r \in P_l$ such that $U \cap \posbody{r} \neq \emptyset$, it follows that $L \cap \posbody{r} \neq \emptyset$.
  From (\ref{eq:notmodel-lfviolated-1}) and using contraposition this means there is some $L \subseteq U$ that is a loop in $P$ and for each $l \in L$ and $r \in P_{l}$ if $L \cap \posbody{r} = \emptyset$ it must hold that $I(\body{r}) < I(l)$.
  Now, for each $l \in L$ and $r \in \nonlooprules{P}{L} \cap P_l$ by definition it holds that $L \cap \supp{I} = \emptyset$, meaning $I(\body{r}) < I(l)$. Thus, $\sup \{ I(\body{r}) \mid r \in \nonlooprules{P}{L} \} < \sup \{ I(l) \mid l \in L \}$, meaning $I \not\models \loopform{L}{P}$.
\end{proof}

Now, we can extend the ASSAT-procedure from \cite{assat-linzhao} to fuzzy answer set programs $P$. The main idea of this method is to use fuzzy SAT solving techniques to find models of the fuzzy propositional theory which consists of the completion of $P$, together with the loop formulas of particular maximal loops of $P$. If a model is found which is not an answer set, then we determine a loop that is violated by the model and add its loop formula to the fuzzy propositional theory, after which the fuzzy SAT solver is invoked again. The algorithm thus becomes:

\begin{enumerate}
 \item Initialize $Loops = \emptyset$
 \item Generate a model $M$ of $\comp{I} \cup \loopform{P}{Loops}$, where $\loopform{P}{Loops}$ is the set of loop formulas of all loops in $Loops$.
 \item If $M = \lfpnfimcons{P^M}$, return $M$ as it is an answer set. Else, find the loops occurring in $\supp{I \fsetminus \lfpnfimcons{P^M}}$, add their loop formulas to $Loops$ and return to step 2.
\end{enumerate}

The reason that we can expect this process to be efficient is articulated by Proposition~\ref{prop:notmodel-lfviolated}. Indeed, when searching for violated loops, we can restrict our attention to subsets of $\supp{I \fsetminus \lfpnfimcons{P^I}}$. Although the worst-case complexity of this algorithm is still exponential, in most practical applications, we can expect $\supp{I \fsetminus \lfpnfimcons{P^I}}$ to be small, as well as the number of iterations of the process that is needed before an answer set is found. In \cite{assat-linzhao} experimental evidence for this claim is provided in the case of classical ASP. Last, note that the fuzzy SAT solving technique depends on the t-norms used in the program. If only the \L ukasiewicz t-norm is used, we can use (bounded) mixed integer programming (bMIP) \cite{ManyValuedMixedInteger}. Since Fuzzy Description Logic Solvers are based on the same techniques as fuzzy SAT solvers, we also know that for the
product t-norm we need to resort to bounded mixed integer quadratically constrained programming (bMICQP) \cite{BobilloStraccia}.



\section{Example: the ATM location selection problem}\label{sec:example}

In this section we illustrate our algorithm on a FASP program modeling a real-life problem. Suppose we are tasked with placing $k$ ATM machines $\mathit{ATM} = \{ a_1,\ldots,a_k \}$ on roads connecting $n$ towns $Towns = \{ t_1,\ldots,t_n \}$ such that the distance between each town and some ATM machine is minimized, i.e.~we aim to find a configuration in which each town has an ATM that is as close as practically possible. To obtain this we optimize the sum of closeness degrees for each town and ATM. Note that this problem closely resembles the well-known $k$-center selection problem (see e.g.~\cite{AusielloAl:ComplexityAndApproximation}). The difference is that in the $k$-center problem the ATMs need to be placed in towns, where we allow them to be placed on the roads connecting towns. We can model this problem as an undirected weighted graph $G = \tuple{V,E}$ where $V = Towns$ is the set of vertices and the edge set $E$ connects two towns if they are directly connected by a road. Given a distance function $d : Towns \times Towns \to \mathbb{R}$ that models the distance between two towns\footnote{For cities that are not connected the function $d$ models the distance of the shortest path between them.}, the weight of the edge $(a,b) \in E$ is given by the normalized distance $d(a,b)/d_{sum}$, where $d_{sum} = \sum\{ d(t_1,t_2) \mid t_1,t_2 \in Towns\}$.

Since our FASP programs can only have t-norms in rule bodies, we also need to find a way to sum up the distances between towns and ATM machines. By using the \emph{nearness degree}, or closeness degree, which for a normalized distance $d$ is defined as $1-d$, we can perform summations of distances in our program. To see this, consider the following derivation:
 \begin{align*}
  \pretnorm_l(1-dist_1,1-dist_2)
    & = \max(1-dist_1 + 1-dist_2-1,0)\\
    & = \max(1-(dist_1 + dist_2),0)\\
    & = 1-\min(dist_1 + dist_2,1) 
 \end{align*}
Hence, by applying the \L ukasiewicz t-norm on the nearness degrees, we are summing the distances.

The program $P_{\mathit{ATM}}$ solving the ATM selection problem is given as follows:

\begin{lprogram}
 \lprule{gloc:}{loc(A,T1,T2)}{\pretnorm_l(conn(T1,T2),\beta)\\
 \lprule{gnear:}{locNear(A,T1)}{\lneg{locNear'(A,T1)}}\\
 \lprule{gnear':}{locNear'(A,T1)}{\pretnorm_l(loc(A,T1,T2),\lneg{near(T1,T2)},}\\
   & & & \phantom{\pretnorm_l(}locNear(A,T2)),T1 \neq T2\\
 \lprule{nearr:}{near(T1,T2)}{\pretnorm_l(conn(T1,T3),near(T1,T3),near(T3,T2))}\\
 \lprule{locr:}{loc(A,T1,T2)}{loc(A,T2,T1)}\\
 \lprule{atmr:}{\mathit{ATMNear}(A,T)}{\pretnorm_l(loc(A,T1,T2),locNear(A,T1),near(T,T1))}\\
 \lprule{tDist:}{\mathit{totNear}}{\pretnorm_l(\{ \mathit{ATMNear}(a,t) \mid a \in \mathit{ATM},t \in \mathit{Towns}\})}\\
\end{lprogram}
where 
 $$\beta = \pretnorm_l(\{\mneg{loc(A,T1',T2')} \mid \{ T'_1,T'_2 \} \neq \{ T_1,T_2 \}\})}$$
Note that, due to grounding, a rule such as $locr$ actually corresponds to a set of variable-free rules $\{ locr_{a,t_1,t_2} \mid a \in ATM, t_1,t_2 \in Towns \}$. We will keep referring to the specific grounded instance of a rule by the subscript.

Program $P_{\mathit{ATM}}$ consists of a \emph{generate} and \emph{define} part, which for a specific configuration is augmented with an \emph{input} part consisting of facts. The \emph{generate} part consists of the three rules $gloc$, $gnear$, and $gnear'$, which generate a specific configuration of ATMs. The $gloc$ rule chooses an edge on which the ATM machine $A$ is placed by guessing a location for an ATM that does not yet has an assigned location, as ensured by the $\beta$ part of this rule. The $gnear$ and $gnear'$ rules generate a location on this edge where $A$ is placed. Rules $gnear$ and $gnear'$ originate from the constraint $d(a,t_1) = d(t_1,t_2) - d(a,t_2)$, where $d(x,y)$ is the distance between $x$ and $y$, if ATM $a$ is placed on the edge between $t_1$ and $t_2$. Defining $n(x,y)$ as the nearness degree between $x$ and $y$ and noting that $n(a,t_1) = 1-d(a,t_1) = 1-(d(t_1,t_2) - d(a,t_2))$, we can rewrite this constraint in terms of t-norms and nearness degrees:
 \begin{align*}
  n(a,t_1)
    &= 1-(d(t_1,t_2) - d(a,t_2))\\
    &= 1-(d(t_1,t_2) + (1-d(a,t_2)) - 1)\\
    &= 1-\pretnorm_l(d(t_1,t_2),1-d(a,t_2))\\
    &= 1-\pretnorm_l(1-n(t_1,t_2),n(a,t_2))\\
    &= \fneg{l}{\pretnorm_l(1-n(t_1,t_2),n(a,t_2))}
 \end{align*}
Hence, the bodies of rules $gnear$ and $gnear'$ ensure that this constraint is satisfied. The reason we need two rules and cannot directly write a rule with body $\fneg{s}{\pretnorm_l(loc(A,T1,T2),\lneg{near(T1,T2)},locNear(A,T2)}$ is that the syntax does not allow negation in front of arbitrary expressions.

Rule $nearr$ recursively defines the degree of closeness between two towns based on the known distances for connected towns.  Additionally, since the bodies of rules with the same head are combined using the maximum, the nearness degree obtained by $nearr$ is always one minus the distance of the shortest path. The $locr$ rule makes sure that if an ATM is located on the edge between town $T1$ and $T2$, it is also recognized as being on the edge between $T2$ and $T1$, as we are working with an undirected graph. The $atmr$ rule defines the location between a particular ATM machine and a town. Note that due to rule $locr$ this rule also covers the case when $near(T,T2)$ is higher than $near(T,T1)$. The $tDist$ rule aggregates the total distances such that different answer sets of this program can be compared and ordered. In this way we could for example search for the answer set that has a maximal total degree of nearness, i.e.~in which the distance from the towns to the ATMs is lowest.

Consider the specific configuration $G_P = \langle V,E \rangle$ of towns $Towns = \{ t_1, t_2, t_3 \}$ depicted in Figure~\ref{fig:atm} and suppose $\mathit{ATM} = \{a_1,a_2\}$. In Figure~\ref{fig:depgraph-atm} we depicted a subset of the dependency graph of the grounded version of $P'_{\mathit{ATM}} = P_{\mathit{ATM}} \cup F$, where $F$ is the input part of the problem, given by the following rules
\begin{align*}
 F = & \{ conn(t,t') \gets 1 \mid t,t' \in Towns, (t,t') \in E \} \\
   \cup & \{ near(t,t') \gets k \mid t,t' \in Towns, (t,t') \in E, k = 1-(d(t,t')/d_{sum}) \}
\end{align*}
For the configuration depicted in Figure~\ref{fig:atm} the input part $F$ is
\begin{align*}
 F = &\{ conn(t_1,t_1) \gets 1, conn(t_1,t_2) \gets 1, conn(t_1,t_3) \gets 1\}\\
& \cup \{ conn(t_2,t_1) \gets 1, conn(t_2,t_2) \gets 1, conn(t_2,t_3) \gets 1 \}\\
& \cup \{ conn(t_3,t_1) \gets 1, conn(t_3,t_2) \gets 1, conn(t_3,t_3) \gets 1 \}\\
& \cup \{ near(t_1,t_1) \gets 1, near(t_1,t_2) \gets 0.8, near(t_1,t_3) \gets 0.7 \}\\
& \cup \{ near(t_2,t_1) \gets 0.8, near(t_2,t_2) \gets 1, near(t_2,t_3) \gets 0.5 \}\\
& \cup \{ near(t_3,t_1) \gets 0.7, near(t_3,t_2) \gets 0.5, near(t_3,t_3) \gets 1 \}
\end{align*}
It is clear that $P'_{\mathit{ATM}}$ contains a number of loops. The completion of $P'_{\mathit{ATM}}$ is the following fuzzy propositional theory:
{\allowdisplaybreaks \begin{align*}
   &conn(t_1,t_1) \feq 1,\;\;\;conn(t_1,t_2) \feq 1,\;\;\;conn(t_1,t_3) \feq 1\\
   &conn(t_2,t_1) \feq 1,\;\;\;conn(t_2,t_2) \feq 1,\;\;\;conn(t_2,t_3) \feq 1\\
   &conn(t_3,t_1) \feq 1,\;\;\;conn(t_3,t_2) \feq 1,\;\;\;conn(t_3,t_3) \feq 1\\
   &near(t_1,t_1) \feq 1,\;\;\;near(t_1,t_2) \feq 0.8,\;\;\;near(t_1,t_3) \feq 0.7\\
   &near(t_2,t_1) \feq 0.8,\;\;\;near(t_2,t_2) \feq 1,\;\;\;near(t_2,t_3) \feq 0.5\\
   &near(t_3,t_1) \feq 0.7,\;\;\;near(t_3,t_2) \feq 0.5,\;\;\;near(t_3,t_3) \feq 1\\
   &loc(a_1,t_1,t_1) \feq_l \max(\pretnorm_l(conn(t_1,t_1),\beta_{1,1,1}),loc(a_1,t_1,t_1))\\
   &loc(a_1,t_1,t_2) \feq_l \max(\pretnorm_l(conn(t_1,t_2),\beta_{1,1,2}),loc(a_1,t_2,t_1))\\
   &loc(a_1,t_1,t_3) \feq_l \max(\pretnorm_l(conn(t_1,t_3,\beta_{1,1,3}),loc(a_1,t_3,t_1))\\
   &\phantom{loc(a_1,t_1,t_3) }\ldots\\
   &loc(a_2,t_3,t_1) \feq_l \max(\pretnorm_l(conn(t_3,t_1),\beta_{2,3,1}),loc(a_2,t_1,t_3))\\
   &loc(a_2,t_3,t_2) \feq_l \max(\pretnorm_l(conn(t_3,t_2),\beta_{2,3,2}),loc(a_2,t_2,t_3))\\
   &loc(a_2,t_3,t_3) \feq_l \max(\pretnorm_l(conn(t_3,t_3),\beta_{2,3,3}),loc(a_2,t_3,t_3))\\
   &locNear(a_1,t_1) \feq_l\; \lneg{locNear'(a_1,t_1)}\\
   &\phantom{locNear(a_1,t_1) }\ldots\\
   &locNear(a_2,t_3) \feq_l \;\lneg{locNear'(a_2,t_3)}\\
   &locNear'(a_1,t_1) \feq_l \max(\pretnorm_l(loc(a_1,t_1,t_2),locNear(a_1,t_2),\lneg{near(t_1,t_2)}),\\
   &\phantom{locNear'(a_1,t_1) \feq_l \max(}\pretnorm_l(loc(a_1,t_1,t_3),locNear(a_1,t_3),\lneg{near(t_1,t_3)}))\\
   &\phantom{locNear'(a_1,t_1) }\ldots\\
   & locNear'(a_2,t_3) \feq_l \max(\pretnorm_l(loc(a_2,t_3,t_1),locNear(a_2,t_1),\lneg{near(t_3,t_1)}),\\
   &\phantom{locNear'(a_2,t_3) \feq_l \max(}\pretnorm_l(loc(a_2,t_3,t_2),locNear(a_2,t_2),\lneg{near(t_3,t_2)}))\\
   &near(t_1,t_1) \feq_l \max(\pretnorm_l(conn(t_1,t_1),near(t_1,t_1),near(t_1,t_1)),\\
   &\phantom{near(t_1,t_1) \feq_l \max(}\pretnorm_l(conn(t_1,t_2),near(t_1,t_2),near(t_2,t_1)),\\
   &\phantom{near(t_1,t_1) \feq_l \max(}\pretnorm_l(conn(t_1,t_3),near(t_1,t_3),near(t_3,t_1)),1)\\
   &near(t_1,t_2) \feq_l \max(\pretnorm_l(conn(t_1,t_1),near(t_1,t_1),near(t_1,t_2)),\\
   &\phantom{near(t_1,t_2) \feq_l \max(}\pretnorm_l(conn(t_1,t_2),near(t_1,t_2),near(t_2,t_2)),\\
   &\phantom{near(t_1,t_2) \feq_l \max(}\pretnorm_l(conn(t_1,t_3),near(t_1,t_3),near(t_3,t_2)),0.8)\\
   &\phantom{near(t_1,t_2) }\ldots\\
   & near(t_3,t_3) \feq_l \max(\pretnorm_l(conn(t_3,t_3),near(t_3,t_3),near(t_3,t_3)),\\
   &\phantom{near(t_3,t_3) \feq_l \max(}\pretnorm_l(conn(t_3,t_2),near(t_3,t_2),near(t_2,t_3)),\\
   &\phantom{near(t_3,t_3) \feq_l \max(}\pretnorm_l(conn(t_3,t_1),near(t_3,t_1),near(t_1,t_3)),1)\\
   &ATMNear(a_1,t_1) \feq \max(\pretnorm_l(loc(a_1,t_1,t_1),locNear(a_1,t_1),near(t_1,t_1)),\\
    &\phantom{ATMNear(a_1,t_1) \feq \max(}\pretnorm_l(loc(a_1,t_1,t_2),locNear(a_1,t_1),near(t_1,t_1)),\\
    &\phantom{ATMNear(a_1,t_1) \feq \max(}\ldots\\
    &\phantom{ATMNear(a_1,t_1) \feq \max(}\pretnorm_l(loc(a_1,t_3,t_2),locNear(a_1,t_3),near(t_1,t_3))\\
    &\phantom{ATMNear(a_1,t_1) \feq \max(}\pretnorm_l(loc(a_1,t_3,t_3),locNear(a_1,t_3),near(t_1,t_3)))\\
    &\phantom{ATMNear(a_1,t_1) }\ldots\\
   &ATMNear(a_2,t_3) \feq \max(\pretnorm_l(loc(a_2,t_1,t_1),locNear(a_2,t_1),near(t_3,t_1)),\\
    &\phantom{ATMNear(a_2,t_3) \feq \max(}\pretnorm_l(loc(a_2,t_1,t_2),locNear(a_2,t_1),near(t_3,t_1)),\\
    &\phantom{ATMNear(a_2,t_3) \feq \max(}\ldots\\
    &\phantom{ATMNear(a_2,t_3) \feq \max(}\pretnorm_l(loc(a_2,t_3,t_2),locNear(a_2,t_3),near(t_3,t_3))\\
    &\phantom{ATMNear(a_2,t_3) \feq \max(}\pretnorm_l(loc(a_2,t_3,t_3),locNear(a_2,t_3),near(t_3,t_3)))\\
   &totNear \feq \pretnorm_l\{ ATMNear(a,t) \mid a \in \mathit{ATM}, t \in Towns \}
\end{align*}}
where 
$$\beta_{i,j,k} = \pretnorm_l(\{\mneg{loc(a_i,t'_j,t'_k)} \mid \{t'_j,t'_k \} \neq \{ t_j,t_k \} \} )$$
Note that e.g.~the $1$ in the right-hand side of the fuzzy proposition with $near(t_1,t_1)$ on the right-hand side stems from the inputs $F$ we added to $P_{\mathit{ATM}}$.
From the completion $\comp{P'_{\mathit{ATM}}}$ we can see that an interpretation $M$ satisfying $M(near(t1,t2)) = 1$ can be a model of $\comp{P'_{\mathit{ATM}}}$, which is clearly unwanted as this would overestimate the nearness degrees between towns (i.e.~underestimate the distances). For example, consider 


\begin{align*}
M = & \{loc(a_1,t_1,t_2)^{1},loc(a_1,t_2,t_1)^1,loc(a_2,t_1,t_3)^{1}, loc(a_2,t_3,t_1)^1,\\
& locNear(a_1,t_1)^{1}, locNear(a_1,t_2)^{1}, locNear(a_2,t_1)^{0.75},locNear(a_2,t_3)^{0.75},\\
& locNear'(a_2,t_1)^{0.25},locNear'(a_2,t_3)^{0.25},near(t_1,t_1)^{1},near(t_1,t_2)^{1},\\
& near(t_2,t_1)^{1}, near(t_1,t_3)^{0.7}, near(t_3,t_1)^{0.7}, near(t_2,t_3)^{0.5}, near(t_3,t_2)^{0.5},\\
& near(t_2,t_2)^{1},near(t_3,t_3)^{1},\mathit{ATMNear(a_1,t_1)}^{1}, \mathit{ATMNear(a_1,t_2)}^{1},\\
& \mathit{ATMNear}(a_1,t_3)^{0.7},\mathit{ATMNear}(a_2,t_1)^{0.75}, \mathit{ATMNear}(a_2,t_2)^{0.75},\\
& \mathit{ATMNear}(a_2,t_3)^{0.75} \}
\end{align*}

Note that atoms $a$ for which $M(a) = 0$ are not included in the set notation, which is e.g.~the case for $totNear$. One can easily verify that $M$ is a model of $\comp{P'_{\mathit{ATM}}}$.
To check whether $M$ is an answer set we compute $\lfpnfimcons{(P'_{\mathit{ATM}})^M}$ by repeatedly applying $\nfimcons{(P'_{\mathit{ATM}})^M}$, starting from the empty set, until we obtain a fixpoint, and check whether $M = \lfpnfimcons{(P'_{\mathit{ATM}})^M}$. Performing this procedure, we obtain

\begin{align*}
\lfpnfimcons{(P'_{\mathit{ATM}})^M} = & \{loc(a_1,t_1,t_2)^{1},loc(a_1,t_2,t_1)^1,loc(a_2,t_1,t_3)^{1}, loc(a_2,t_3,t_1)^1,\\
& locNear(a_1,t_1)^{1}, locNear(a_1,t_2)^{1}, locNear(a_2,t_1)^{0.75},\\
& locNear(a_2,t_3)^{0.75},locNear'(a_2,t_1)^{0.25},locNear'(a_2,t_3)^{0.25},\\
& near(t_1,t_1)^{1},near(t_1,t_2)^{0.8},near(t_2,t_1)^{0.8}, near(t_1,t_3)^{0.5},\\
& near(t_3,t_1)^{0.5},near(t_2,t_3)^{0.7}, near(t_3,t_2)^{0.7},near(t_2,t_2)^{1},near(t_3,t_3)^{1},\\
& \mathit{ATMNear(a_1,t_1)}^{1},\mathit{ATMNear(a_1,t_2)}^{1},\mathit{ATMNear}(a_1,t_3)^{0.7},\\
& \mathit{ATMNear}(a_2,t_1)^{0.75},\mathit{ATMNear}(a_2,t_2)^{0.75},\mathit{ATMNear}(a_2,t_3)^{0.75} \}
\end{align*}

We can see that $\lfpnfimcons{(P'_{\mathit{ATM}})^M}(near(t_1,t_2)) = 0.8 \neq M(near(t_1,t_2))$, hence $M$ is not an answer set of $P'_{\mathit{ATM}}$.
From Proposition~\ref{prop:notmodel-lfviolated} we then know that there must be a loop in $\supp{M \fsetminus \lfpnfimcons{(P'_{\mathit{ATM}})^M}} = \{near(t_1,t_2), near(t_2,t_1)\}$ whose loop formula is violated. Looking at the dependency graph, we can see that $L = \supp{M \fsetminus \lfpnfimcons{(P'_{\mathit{ATM}})^M}} = \{near(t_1,t_2), near(t_2,t_1)\}$ contains three loops: $L_1 = L$, $L_2 = \{ near(t_1,t_2) \}$ and $L_3 = \{ near(t_2,t_1) \}$. Their loop formulas are
 \begin{align*}
  & \loopform{L_1}{P'_{\mathit{ATM}}} = \prefimp(\max\Big(near(t_1,t_2),near(t_2,t_1)\Big),\max\Big(\pretnorm_l(conn(t_1,t_3),\\
  & \phantom{xxxxxx}near(t_1,t_3),near(t_3,t_2)),0.8,\pretnorm_l(conn(t_2,t_3),near(t_2,t_3),near(t_3,t_1))\Big)\\
  & \loopform{L_2}{P'_{\mathit{ATM}}} = \\
  & \phantom{xxxx}\prefimp(\max\Big(near(t_1,t_2)\Big),\max\Big(\pretnorm_l(conn(t_1,t_3),near(t_1,t_3),near(t_3,t_2)),0.8\Big)\\
  & \loopform{L_3}{P'_{\mathit{ATM}}} = \\
  & \phantom{xxxx}\prefimp(\max\Big(near(t_2,t_1)\Big),\max\Big(\pretnorm_l(conn(t_2,t_3),near(t_2,t_3),near(t_3,t_1)),0.8\Big)
 \end{align*}
 Clearly, these loop formulas are violated by $M$, hence following the algorithm introduced in
Section~\ref{sec:loopelimination}, we create a new fuzzy propositional theory $\comp{P'_{\mathit{ATM}}} \cup
\{\loopform{L_1}{P'_{\mathit{ATM}}},\loopform{L_2}{P'_{\mathit{ATM}}},\loopform{L_3}{P'_{\mathit{ATM}}}\}$, and try to
find a model of this new theory. Consider then the following model of this new theory:

\begin{align*}
M = & \{loc(a_1,t_1,t_2)^{1},loc(a_1,t_2,t_1)^1,loc(a_2,t_1,t_3)^{1}, loc(a_2,t_3,t_1)^1,\\
& locNear(a_1,t_1)^{0.15}, locNear(a_1,t_2)^{0.05},locNear'(a_1,t_1)^{0.85},\\
& locNear'(a_1,t_2)^{0.95}locNear(a_2,t_1)^{0.75},locNear(a_2,t_3)^{0.75},\\
& locNear'(a_2,t_1)^{0.25},locNear'(a_2,t_3)^{0.25},near(t_1,t_1)^{1},near(t_1,t_2)^{0.8}, \\
& near(t_2,t_1)^{0.8},near(t_1,t_3)^{0.7},near(t_3,t_1)^{0.7},near(t_2,t_3)^{0.5}, near(t_3,t_2)^{0.5},\\
& near(t_2,t_2)^{1},near(t_3,t_3)^{1},\mathit{ATMNear(a_1,t_1)}^{0.85}, \mathit{ATMNear(a_1,t_2)}^{0.95},\\
& \mathit{ATMNear}(a_1,t_3)^{0.55},\mathit{ATMNear}(a_2,t_1)^{0.75}, \mathit{ATMNear}(a_2,t_2)^{0.55},\\
& \mathit{ATMNear}(a_2,t_3)^{0.75} \}
\end{align*}

One can readily verify that this model is an answer set of $P'_{\mathit{ATM}}$, hence the algorithm stops
and returns $M$.

\begin{figure}
 \centering
 \begin{tikzpicture}[scale=2.5,auto,swap]
 \path (0:0cm) node[draw,shape=circle] (v1) {$t_1$};
 \path (-18:1.5cm) node[draw,shape=circle] (v2) {$t_2$};
 \path (18:1.3cm) node[draw,shape=circle] (v0) {$t_3$};
 \draw (v0) -- node {$0.7$} (v1) (v0) -- node {$0.5$} (v2) (v1) -- node {$0.8$}(v2);
 \end{tikzpicture}
 \caption{Town configuration for $P_{ATM}$. The weights on the edges denote the nearness degrees between towns $t_1$, $t_2$ and $t_3$}
 \label{fig:atm}
\end{figure}
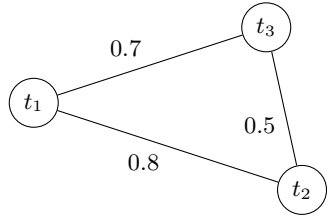

\begin{figure}
 \centering
 \includegraphics[scale=0.25]{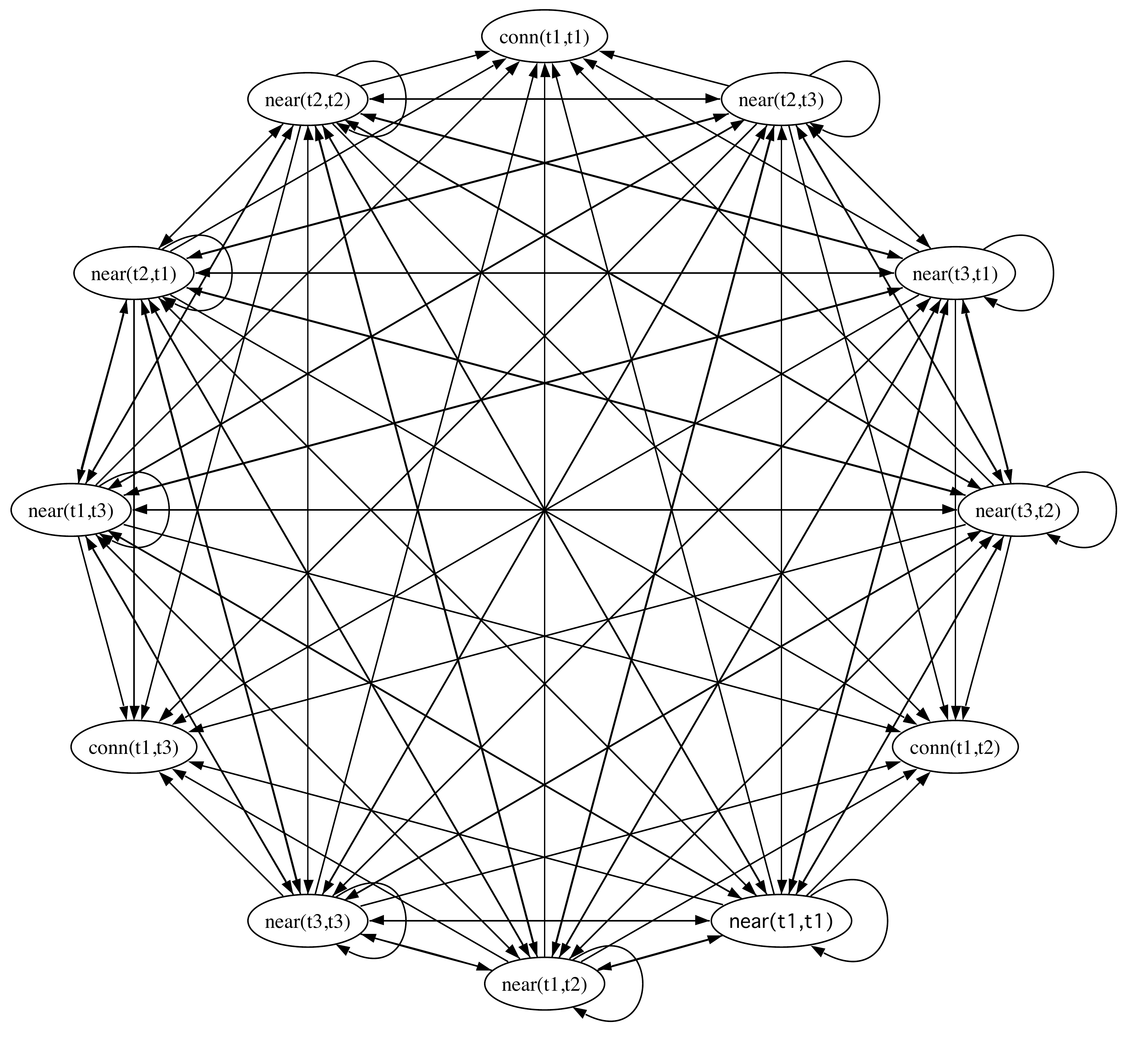}
 \caption{Dependency graph of $P_{ATM}$}
 \label{fig:depgraph-atm}
\end{figure}

We could have solved this problem using Mixed Integer Programming (MIP)\footnote{Though in general the G\"odel negation $\mathcal{N}_m$ cannot be implemented in MIP, in the ATM example we can implement the $\mathit{gloc}$ rules using integer variables.}. However, the exact encoding of this problem would be less clear and straightforward to write. The reason for this is that in the MIP translation the loop formulas would need to be explicitly represented in the program, while in FASP this is handled implicitly. Hence, only the implementer of a FASP system needs to handle these loop formulas, not the developer who writes the FASP programs. This is exactly the power of FASP: providing an elegant, concise, and clear modelling language for representing continuous problems, which, thanks to the results in this paper, can be automatically translated to lower-level languages for solving continuous problems, such as MIP.

\section{Discussion}\label{sec:restrict}

The reader might wonder why we limit our approach to FASP programs with t-norms in their body, because at first sight it
seems the presented approach is easily extendable to arbitrary functions. It turns out that this is not the case,
however. Consider FASP with the \L ukasiewicz t-norm in rule bodies. As mentioned before, the completion of such a program, and its loop formulas, are formulas in \L ukasiewicz logic and are implementable using MIP. Now let us consider FASP where both the \L ukasiewicz t-norm and the \L ukasiewicz t-conorm may occur in rule bodies. At first, one would suspect that the loop formulas of such a program would again be formulas in \L ukasiewicz logic.
This turns out to be wrong however. To see this, consider the following rules:
 \begin{align*}
  b &\gets \lneg{a}\\
  b &\gets \pretconorm_l(b,b)
 \end{align*}
One can readily verify that in the answer sets of a program containing these rules, literal $b$ will be equal to $\mneg{a}$ (provided that $b$ does not occur in the head of any other rule). 
However, the negation $\mathcal{N}_m$ cannot be implemented in MIP, as the solution space of a MIP problem is always a topologically closed set (viz.~the union of a finite number of polyhedra), whereas the solution space of a constraint $b \feq \mneg{a}$ cannot be represented as a closed set due to the strict negation in the definition of $\mathcal{N}_m$. This means that as soon as the \L ukasiewicz t-conorm is allowed, in general, there will not exist a \L ukasiewicz logic theory such that the models of that theory coincide with the answer sets of a given program. Hence, it is clear that the case where other operators than t-norms are used requires a different strategy.

Finding generalized loop formulas that cover e.g. both the \L ukasiewicz t-norm and t-conorm is not a trivial problem.  To illustrate some of the issues, let us examine two intuitive candidates. First, remark that the loop formulas introduced in Section~\ref{sec:loopelimination} eliminate certain answer sets (i.e.~they are too strict). Consider the following program $P$:

\begin{align*}
 a &\gets \pretconorm_l(a,b)\\
 b &\gets k
\end{align*}

\noindent where $k \in [0,1]$. This program has one loop, viz.~$\{a,b\}$ with corresponding loop formula $\max(a,b) \leq k$.
Now note that for $k > 0$ the value of $a$ in any answer set is equal to $1$. Hence, the loop formula incorrectly eliminates all answer sets in this case. One might think this can be solved by including a condition in the loop formula: $(\max(a,b) \leq l) \vee (b > 0)$. This formula however fails to eliminate models that are not answer sets (i.e.~it is not strict enough) on the following program:

\begin{align*}
 a &\gets \pretnorm_m(\pretconorm_l(a,b),0.8)\\
 b &\gets \pretconorm_l(a,b)\\
 b &\gets k
\end{align*}

If $k > 0$ the unique answer set of this program is $\{a^{0.8},b^{1}\}$. However, $\{a^1,b^1\}$ is also a model of the completion
of this program and satisfies the above loop formula.

Although again more refined loop formulas can be thought of that handle the latter program correctly, we are pessimistic about the possibility of finding loop formulas that cover all cases.   It appears that such a general solution should be able to capture some underlying idea of recursion:  one loop may justify the truth value of some atom a, up to a certain level, which may then trigger other rules that justify the truth value of a, up to some higher level, etc.

Note that this problem does not occur in classical ASP (or when using the maximum t-conorm), since e.g. $a \gets b \vee c$ is equivalent to $a \gets b$ and $a \gets c$, which is indeed why disjunctions in the body of rules are not considered in classical ASP. 

\section{Related Work}\label{sec:related}

The approach to fuzzy answer set programming for which we provided the translation to fuzzy SAT is called an \emph{unweighted implication-based approach}. There also exist \emph{weighted implication-based approaches} (e.g.~\cite{LukaStraccia07,madrid:existenceandunicity,MadridOjeda-Aciego-2008a,MadridAciego2009}), which use rules of the form
 \begin{equation}r: a \stackrel{\alpha}{\gets} \pretnorm(b_{1},\ldots,b_{n})\label{eq:wrule}\end{equation}
 where $r$ is a rule label, $a$ is an atom, $b_{i}$, for $1 \leq i \leq n$, are extended literals, and $\alpha \in [0,1]$. An interpretation $I$ models this rule iff
 $$\rulefimp{r}(\pretnorm(I(b_{1}),\ldots,I(b_{n})),I(a)) \geq \alpha$$
 Since $\rulefimp{r}$ is the residual implicator of $\pretnorm$ this is equivalent to 
 $$I(a) \geq \pretnorm(I(b_{1}),\ldots,I(b_{n}),\alpha)$$
 Hence a weighted rule of the form (\ref{eq:wrule}) above can be simulated by the rule
 $$r': a \gets \pretnorm(b_{1},\ldots,b_{n},\alpha)$$
Thus, the translation presented in this paper can equally be applied to weighted implication-based approaches.

In addition to the implication-based approaches (IB) one also finds \emph{annotation-based (AB) approaches} (see e.g.~\cite{straccia-annotated}). In the annotation-based setting a rule is of the form 
 $$A: f(\beta_{1},\ldots,\beta_{n}) \gets B_{1} : \beta_{1},\ldots,B_{n} : \beta_{n}$$
Such a rule asserts that the value of atom $A$ is at least $f(\beta_{1},\ldots,\beta_{n})$ if the value of each atom $B_{i}$, $1 \leq n$, is at least $\beta_{i}$. In this setting $f$ is a computable function and $\beta_{i}$ is either a constant or a variable ranging over an appropriate truth domain. Due to the difference in semantics between the IB and AB approaches, our method is not directly applicable to AB frameworks. One can find an in-depth overview of logic programming with fuzzy logic in~\cite{Straccia:reasoningweb}.

In \cite{fasp1}, an implementation method for FASP programs with a finite truth value set is presented, which consists of translating a FASP program to a specific DLVHEX program. For solving continuous problems, however, we need infinite truth values, for which a solving method is much harder to construct. Our method is able to handle continuous problems, and additionally is more flexible than \cite{fasp1} since any method for solving continuous problems can be used as the backend, including fuzzy SAT solvers and the vast body of existing MIP solvers.

Apart from fuzzy answer set programming, in recent years possibilistic and probabilistic answer set programming have been developed. Both of these approaches can be reduced to classical SAT. In the case of probabilistic ASP, there is a direct translation method \cite{Saad:ECSQARU2009}, while a possibilistic ASP program can be translated to an equivalent ASP program, on which the ASSAT procedure can then be applied.

\section{Conclusion}\label{sec:conclusion}

In this paper we have focused on the translation of FASP programs to particular satisfiability problems. We have introduced the completion of a program and have shown that in the case of programs without loops, the models of the completion are exactly the answer sets. Furthermore, to solve the general problem, we have generalized the notion of loop formulas. This translation is important because it allows to solve FASP programs using fuzzy SAT solvers. Under appropriate restrictions, for example, the satisfiability problems that are obtained can be solved using off-the-shelf mixed integer programming methods.  From an application point of view, this allows us to encode continuous optimization problems in a declarative style which is similar to traditional answer set programming.   This style of encoding problems is often more intuitive, as well as more concise, while the results we have presented ensure that the power of mathematical programming techniques can still be employed to find the solutions.

\section*{Acknowledgment}
The authors would like to thank the anonymous reviewers for their useful suggestions and remarks.
\bibliographystyle{acmtrans}
\bibliography{loops}

\end{document}

%% file: dvasp.tex
\newenvironment{lprogram}{\[\begin{array}{rrll}}{\end{array}\]}

\newcommand{\lprule}[3]{\ensuremath{#1 & #2 &\gets & #3}}

\newcommand{\hbase}[1]{{\ensuremath{{\cal B}_{#1}}}}

\newcommand{\body}[1]{\ensuremath{B(#1)}}
\newcommand{\head}[1]{\ensuremath{H(#1)}}

\newcommand{\comp}[1]{\ensuremath{comp(#1)}}
\newcommand{\poslit}[1]{\ensuremath{Lit^{+}(#1)}}
\newcommand{\posbody}[1]{\ensuremath{\poslit{\body{#1}}}}
\newcommand{\depgraph}[1]{\ensuremath{G_{#1}}}
\newcommand{\looprules}[2]{R^{+}_{#1}(#2)}
\newcommand{\nonlooprules}[2]{R^{-}_{#1}(#2)}
\newcommand{\loopform}[2]{\mathbb{LF}(#1,#2)}
\newcommand{\loopformprog}[1]{\mathbb{LF}(#1)}

%% file: dvfuzzy.tex
\newcommand{\Fuzzy}[1]{\ensuremath{\mathcal{F}({#1})}}

\newcommand{\supp}[1]{\ensuremath{\mathit{supp}(#1)}}

\newcommand{\fneg}[2]{\ensuremath{\mathcal{N}_{#1}(#2)}}
\newcommand{\lneg}[1]{\fneg{l}{#1}}
\newcommand{\mneg}[1]{\fneg{m}{#1}}
\newcommand{\pneg}[1]{\fneg{p}{#1}}

\newcommand{\feq}{\ensuremath{\approx}}
\newcommand{\bodyfand}[1]{\pretnorm_{#1}}

\newcommand{\rulefimp}[1]{\prefimp_{#1}}

\newcommand{\fsetminus}{\circleddash}
\newcommand{\pretnorm}{\mathcal{T}}
\newcommand{\pretconorm}{\mathcal{S}}
\newcommand{\prefimp}{\mathcal{I}}
\newcommand{\nfimcons}[1]{\Pi_{#1}}
\newcommand{\lfpnfimcons}[1]{\nfimcons{#1}^{*}}

%% file: general.tex
\newcommand{\bemph}[1]{\textbf{#1}}

\newcommand{\tuple}[1]{\ensuremath{\langle{#1}\rangle}}

\DeclareMathSymbol{\FORALL}   {\mathord}{symbols}{"38}
\DeclareMathSymbol{\EXISTS}   {\mathord}{symbols}{"39}
\newcommand{\SUCHTHAT}{\colon}
\newcommand{\Exists}[2]{\ensuremath{\EXISTS{#1} \SUCHTHAT {#2}}}
\newcommand{\Forall}[2]{\ensuremath{\FORALL{#1} \SUCHTHAT {#2}}}

  \newcommand{\lfpimcons}[2]{\esm{\Pi^*_{#1,#2}}}

 \newtheorem{definition}{Definition}
 \newtheorem{example}{Example}
 \newtheorem{proposition}{Proposition}
 
 \newtheorem{lemma}{Lemma}

%% file: jjgeneral.tex
 \newcommand{\esm}[1]{\ensuremath{#1}}	

 \newcommand{\imp}{\Rightarrow}		

  \newcommand{\lub}{\esm{\sqcup}}	
  \newcommand{\glb}{\esm{\sqcap}}	

 \newcommand{\hint}[1]{\langle \textnormal{#1} \rangle\;}
 

%% file: loops_CoRR.bbl
\begin{thebibliography}{}

\bibitem[\protect\citeauthoryear{Alsinet, Godo, and Sandri}{Alsinet
  et~al\mbox{.}}{2002}]{Alsinet:possibilistic}
{\sc Alsinet, T.}, {\sc Godo, L.}, {\sc and} {\sc Sandri, S.} 2002.
\newblock Two formalisms of extended possibilistic logic programming with
  context-dependent fuzzy unification: A comparative description.
\newblock {\em Electronic Notes in Theoretical Computer Science\/}~{\em
  66,\/}~5, 1 -- 21.

\bibitem[\protect\citeauthoryear{Ausiello, Crescenzi, Gambosi, Kann,
  Marchetti-Spaccamela, and Protasi}{Ausiello
  et~al\mbox{.}}{1999}]{AusielloAl:ComplexityAndApproximation}
{\sc Ausiello, G.}, {\sc Crescenzi, P.}, {\sc Gambosi, G.}, {\sc Kann, V.},
  {\sc Marchetti-Spaccamela, A.}, {\sc and} {\sc Protasi, M.} 1999.
\newblock {\em Complexity and Approximation}.
\newblock Springer-Verlag.

\bibitem[\protect\citeauthoryear{Baral}{Baral}{2003}]{BaralBook}
{\sc Baral, C.} 2003.
\newblock {\em Knowledge Representation, Reasoning and Declarative Problem
  Solving}.
\newblock Cambridge University Press.

\bibitem[\protect\citeauthoryear{Baral, Gelfond, and Rushton}{Baral
  et~al\mbox{.}}{2007}]{baral:plog}
{\sc Baral, C.}, {\sc Gelfond, M.}, {\sc and} {\sc Rushton, N.} 2007.
\newblock Probabilistic reasoning with answer sets.
\newblock In {\em Proceedings of the 9th International Conference on Logic
  Programming and Nonmonotonic Reasoning (LPNMR'07)}, {V.~Lifschitz} {and}
  {I.~Niemel{\"a}}, Eds. LNCS, vol. 2923. Springer Berlin / Heidelberg, 21--33.

\bibitem[\protect\citeauthoryear{Bauters, Schockaert, {De Cock}, and
  Vermeir}{Bauters et~al\mbox{.}}{2010}]{Kim:UAI2010}
{\sc Bauters, K.}, {\sc Schockaert, S.}, {\sc {De Cock}, M.}, {\sc and} {\sc
  Vermeir, D.} 2010.
\newblock Possibilistic answer set programming revisited.
\newblock In {\em Proceedings of the 26th Conference on Uncertainty in
  Artificial Intelligence (UAI-10)}, {P.~Gr\"{u}nwald} {and} {P.~Spirtes}, Eds.
  AUAI Press.

\bibitem[\protect\citeauthoryear{Bobillo and Straccia}{Bobillo and
  Straccia}{2007}]{BobilloStraccia}
{\sc Bobillo, F.} {\sc and} {\sc Straccia, U.} 2007.
\newblock A fuzzy description logic with product t-norm.
\newblock In {\em Proceedings of the 16th IEEE International Conference on
  Fuzzy Systems (FUZZ-IEEE 2007)}. IEEE Computer Society, 652--657.

\bibitem[\protect\citeauthoryear{Cao}{Cao}{2000}]{Cao2000}
{\sc Cao, T.~H.} 2000.
\newblock Annotated fuzzy logic programs.
\newblock {\em Fuzzy Sets \& Systems\/}~{\em 113,\/}~2, 277--298.

\bibitem[\protect\citeauthoryear{Dam\'asio, Medina, and Ojeda-Aciego}{Dam\'asio
  et~al\mbox{.}}{2004}]{damasio:sortedmultiadjoint}
{\sc Dam\'asio, C.~V.}, {\sc Medina, J.}, {\sc and} {\sc Ojeda-Aciego, M.}
  2004.
\newblock Sorted multi-adjoint logic programs: termination results and
  applications.
\newblock In {\em Proceedings of the 9th European Conference on Logics in
  Artificial Intelligence (JELIA'04)}, {J.~J. Alferes} {and} {J.~Leite}, Eds.
  LNCS, vol. 3229. Springer Berlin / Heidelberg, 252--265.

\bibitem[\protect\citeauthoryear{Dam{\'a}sio, Medina, and
  Ojeda-Aciego}{Dam{\'a}sio et~al\mbox{.}}{2007}]{DamasioMO07}
{\sc Dam{\'a}sio, C.~V.}, {\sc Medina, J.}, {\sc and} {\sc Ojeda-Aciego, M.}
  2007.
\newblock Termination of logic programs with imperfect information:
  applications and query procedure.
\newblock {\em Journal of Applied Logic\/}~{\em 5,\/}~3, 435--458.

\bibitem[\protect\citeauthoryear{Dam\'{a}sio and Pereira}{Dam\'{a}sio and
  Pereira}{2000}]{damasio-hybridprobabilistic}
{\sc Dam\'{a}sio, C.~V.} {\sc and} {\sc Pereira, L.~M.} 2000.
\newblock Hybrid probabilistic logic programs as residuated logic programs.
\newblock In {\em Proceedings of the 7th European Workshop on Logics in
  Artificial Intelligence (JELIA'00)}, {M.~Ojeda-Aciego}, {I.~de~Guzm{\'a}n},
  {G.~Brewka}, {and} {L.~Moniz~Pereira}, Eds. LNCS, vol. 1919. Springer Berlin
  / Heidelberg, 57--72.

\bibitem[\protect\citeauthoryear{Dam\'{a}sio and Pereira}{Dam\'{a}sio and
  Pereira}{2001a}]{damasio-antitonic}
{\sc Dam\'{a}sio, C.~V.} {\sc and} {\sc Pereira, L.~M.} 2001a.
\newblock Antitonic logic programs.
\newblock In {\em Proceedings of the 6th International Conference on Logic
  Programming and Nonmonotonic Reasoning (LPNMR'01)}, {T.~Eiter}, {W.~Faber},
  {and} {M.~Truszczynski}, Eds. LNCS, vol. 2173. Springer Berlin / Heidelberg,
  379--393.

\bibitem[\protect\citeauthoryear{Dam\'{a}sio and Pereira}{Dam\'{a}sio and
  Pereira}{2001b}]{DamasioViegasPereira2001}
{\sc Dam\'{a}sio, C.~V.} {\sc and} {\sc Pereira, L.~M.} 2001b.
\newblock Monotonic and residuated logic programs.
\newblock In {\em Proceedings of the 6th European Conference on Symbolic and
  Quantitative Approaches to Reasoning with Uncertainty (ECSQARU'01)},
  {S.~Benferhat} {and} {P.~Besnard}, Eds. LNCS, vol. 2143. Springer Berlin /
  Heidelberg, 748--759.

\bibitem[\protect\citeauthoryear{Dam\'{a}sio and Pereira}{Dam\'{a}sio and
  Pereira}{2004}]{DamasioPereira-embeddings}
{\sc Dam\'{a}sio, C.~V.} {\sc and} {\sc Pereira, L.~M.} 2004.
\newblock Sorted monotonic logic programs and their embeddings.
\newblock In {\em Proceedings of Information Processing and Management of
  Uncertainty (IPMU04)}. 807--814.

\bibitem[\protect\citeauthoryear{Davis and Putnam}{Davis and
  Putnam}{1960}]{dpll}
{\sc Davis, M.} {\sc and} {\sc Putnam, H.} 1960.
\newblock A computing procedure for quantification theory.
\newblock {\em Journal of the ACM\/}~{\em 7,\/}~3, 201--215.

\bibitem[\protect\citeauthoryear{Emden}{Emden}{1986}]{vanEmden1986}
{\sc Emden, M. H.~v.} 1986.
\newblock Quantitative deduction and its fixpoint theory.
\newblock {\em Journal of Logic Programming\/}~{\em 30,\/}~1, 37--53.

\bibitem[\protect\citeauthoryear{Fages}{Fages}{1994}]{fages:completion}
{\sc Fages, F.} 1994.
\newblock Consistency of {C}lark's completion and existence of stable models.
\newblock {\em Methods of Logic in Computer Science\/}~{\em 1}, 51--60.

\bibitem[\protect\citeauthoryear{Fitting}{Fitting}{1991}]{Fitting1991}
{\sc Fitting, M.} 1991.
\newblock Bilattices and the semantics of logic programming.
\newblock {\em Journal of Logic Programming\/}~{\em 11,\/}~2, 91--116.

\bibitem[\protect\citeauthoryear{Fuhr}{Fuhr}{2000}]{Fuhr2000}
{\sc Fuhr, N.} 2000.
\newblock Probabilistic datalog: implementing logical information retrieval for
  advanced applications.
\newblock {\em Journal of the American Society for Information Science\/}~{\em
  51,\/}~2, 95--110.

\bibitem[\protect\citeauthoryear{Gebser, Kaufmann, and Schaub}{Gebser
  et~al\mbox{.}}{2009}]{clasp}
{\sc Gebser, M.}, {\sc Kaufmann, B.}, {\sc and} {\sc Schaub, T.} 2009.
\newblock The conflict-driven answer set solver clasp: Progress report.
\newblock In {\em Proceedings of the 10th International Conference on Logic
  Programming and Nonmonotonic Reasoning (LPNMR'09)}, {E.~Erdem}, {F.~Lin},
  {and} {T.~Schaub}, Eds. LNCS, vol. 5753. Springer Berlin / Heidelberg,
  509--514.

\bibitem[\protect\citeauthoryear{Gelfond and Lifschitz}{Gelfond and
  Lifschitz}{1988}]{gelfondl88}
{\sc Gelfond, M.} {\sc and} {\sc Lifschitz, V.} 1988.
\newblock The stable model semantics for logic programming.
\newblock In {\em Proceedings of the Fifth International Conference and
  Symposium on Logic Programming (ICLP/SLP'88)}. MIT Press, 1081--1086.

\bibitem[\protect\citeauthoryear{Giunchiglia, Lierler, and Maratea}{Giunchiglia
  et~al\mbox{.}}{2004}]{cmodels}
{\sc Giunchiglia, E.}, {\sc Lierler, Y.}, {\sc and} {\sc Maratea, M.} 2004.
\newblock {SAT}-based answer set programming.
\newblock In {\em Proceedings of the 19th national conference on Artifical
  intelligence (AAAI'04)}. AAAI Press / The MIT Press, 61--66.

\bibitem[\protect\citeauthoryear{H\"ahnle}{H\"ahnle}{1994}]{ManyValuedMixedInt%
eger}
{\sc H\"ahnle, R.} 1994.
\newblock Many-valued logic and mixed integer programming.
\newblock {\em Annals of Mathematics and Artificial Intelligence\/}~{\em
  12,\/}~3-4, 231--263.

\bibitem[\protect\citeauthoryear{H\'{a}jek}{H\'{a}jek}{2001}]{Hajek98}
{\sc H\'{a}jek, P.} 2001.
\newblock {\em Metamathematics of Fuzzy Logic (Trends in Logic)}.
\newblock Springer.

\bibitem[\protect\citeauthoryear{Ishizuka and Kanai}{Ishizuka and
  Kanai}{1985}]{IshizukaMitsuru-1985a}
{\sc Ishizuka, M.} {\sc and} {\sc Kanai, N.} 1985.
\newblock Prolog-{ELF} incorporating fuzzy logic.
\newblock In {\em Proceedings of the 9th international joint conference on
  Artificial intelligence (IJCAI'85)}. 701--703.

\bibitem[\protect\citeauthoryear{Janssen, Heymans, Vermeir, and {De
  Cock}}{Janssen et~al\mbox{.}}{2008}]{FASP-iclp08}
{\sc Janssen, J.}, {\sc Heymans, S.}, {\sc Vermeir, D.}, {\sc and} {\sc {De
  Cock}, M.} 2008.
\newblock Compiling fuzzy answer set programs to fuzzy propositional theories.
\newblock In {\em Proceedings of the 24th International Conference on Logic
  Programming (ICLP'08)}, {M.~Garcia de~la Banda} {and} {E.~Pontelli}, Eds.
  LNCS, vol. 5366. Springer Berlin / Heidelberg, 362--376.

\bibitem[\protect\citeauthoryear{Kifer and Li}{Kifer and
  Li}{1988}]{KiferLi-1988a}
{\sc Kifer, M.} {\sc and} {\sc Li, A.} 1988.
\newblock On the semantics of rule-based expert systems with uncertainty.
\newblock In {\em Proceedings of the 2nd International Conference on Database
  Theory (ICDT'88)}, {M.~Gyssens}, {J.~Paredaens}, {and} {D.~{Van Gucht}}, Eds.
  LNCS, vol. 326. Springer Berlin / Heidelberg, 102--117.

\bibitem[\protect\citeauthoryear{Kifer and Subrahmanian}{Kifer and
  Subrahmanian}{1992}]{kifer92theory}
{\sc Kifer, M.} {\sc and} {\sc Subrahmanian, V.~S.} 1992.
\newblock Theory of generalized annotated logic programming and its
  applications.
\newblock {\em Journal of Logic Programming\/}~{\em 12,\/}~3\&4, 335--367.

\bibitem[\protect\citeauthoryear{Lakshmanan}{Lakshmanan}{1994}]{Lakshmanan1994}
{\sc Lakshmanan, L. V.~S.} 1994.
\newblock An epistemic foundation for logic programming with uncertainty.
\newblock In {\em Proceedings of the 14th Conference on Foundations of Software
  Technology and Theoretical Computer Science (FSTTCS'94)}, {P.~Thiagarajan},
  Ed. LNCS, vol. 880. Springer Berlin / Heidelberg, 89--100.

\bibitem[\protect\citeauthoryear{Lakshmanan}{Lakshmanan}{1997}]{PhDLakshmanan}
{\sc Lakshmanan, L. V.~S.} 1997.
\newblock Towards a generalized theory of deductive databases with uncertainty.
\newblock Ph.D. thesis, Concordia University.

\bibitem[\protect\citeauthoryear{Lakshmanan and Sadri}{Lakshmanan and
  Sadri}{1994}]{LakshmananSadri-1994}
{\sc Lakshmanan, L. V.~S.} {\sc and} {\sc Sadri, F.} 1994.
\newblock Modeling uncertainty in deductive databases.
\newblock In {\em Proceedings of the 5th International Conference on Database
  and Expert Systems Applications (DEXA'94)}, {D.~Karagiannis}, Ed. LNCS, vol.
  856. Springer Berlin / Heidelberg, 724--733.

\bibitem[\protect\citeauthoryear{Lakshmanan and Sadri}{Lakshmanan and
  Sadri}{1997}]{LakshmananSadri-1997}
{\sc Lakshmanan, L. V.~S.} {\sc and} {\sc Sadri, F.} 1997.
\newblock Uncertain deductive databases: a hybrid approach.
\newblock {\em Information Systems\/}~{\em 22,\/}~9, 483--508.

\bibitem[\protect\citeauthoryear{Lakshmanan and Shiri}{Lakshmanan and
  Shiri}{2001}]{LakshmananSadri-2001}
{\sc Lakshmanan, L. V.~S.} {\sc and} {\sc Shiri, N.} 2001.
\newblock A parametric approach to deductive databases with uncertainty.
\newblock {\em IEEE Transactions on Knowledge and Data Engineering\/}~{\em
  13,\/}~4, 554--570.

\bibitem[\protect\citeauthoryear{Leone, Pfeifer, Faber, Eiter, Gottlob, Perri,
  and Scarcello}{Leone et~al\mbox{.}}{2006}]{dlv}
{\sc Leone, N.}, {\sc Pfeifer, G.}, {\sc Faber, W.}, {\sc Eiter, T.}, {\sc
  Gottlob, G.}, {\sc Perri, S.}, {\sc and} {\sc Scarcello, F.} 2006.
\newblock The {DLV} system for knowledge representation and reasoning.
\newblock {\em ACM Transactions on Computational Logic\/}~{\em 7,\/}~3,
  499--562.

\bibitem[\protect\citeauthoryear{Lin and Zhao}{Lin and
  Zhao}{2004}]{assat-linzhao}
{\sc Lin, F.} {\sc and} {\sc Zhao, Y.} 2004.
\newblock {ASSAT}: computing answer sets of a logic program by {SAT} solvers.
\newblock {\em Artificial Intelligence\/}~{\em 157,\/}~1-2, 115--137.

\bibitem[\protect\citeauthoryear{Liu and Truszczy\'nski}{Liu and
  Truszczy\'nski}{2005}]{pbmodels}
{\sc Liu, L.} {\sc and} {\sc Truszczy\'nski, M.} 2005.
\newblock Pbmodels -- software to compute stable models by pseudoboolean
  solvers.
\newblock In {\em Proceedings of the 8th international conference on Logic
  Programming and Nonmonotonic Reasoning (LPNMR'05)}, {C.~Baral}, {G.~Greco},
  {N.~Leone}, {and} {G.~Terracina}, Eds. LNCS, vol. 3662. Springer Berlin /
  Heidelberg, 410--415.

\bibitem[\protect\citeauthoryear{Loyer and Straccia}{Loyer and
  Straccia}{2002}]{LoyerStraccia-2002a}
{\sc Loyer, Y.} {\sc and} {\sc Straccia, U.} 2002.
\newblock The well-founded semantics in normal logic programs with uncertainty.
\newblock In {\em Proceedings of the 6th International Symposium on Functional
  and Logic Programming (FLOPS'02)}, {Z.~Hu} {and}
  {M.~Rodr\'{\i}guez-Artalejo}, Eds. LNCS, vol. 2441. Springer, 152--166.

\bibitem[\protect\citeauthoryear{Loyer and Straccia}{Loyer and
  Straccia}{2003}]{LoyerStraccia-2003}
{\sc Loyer, Y.} {\sc and} {\sc Straccia, U.} 2003.
\newblock The approximate well-founded semantics for logic programs with
  uncertainty.
\newblock In {\em Proceedings of the 28th International Symposium on
  Mathematical Foundations of Computer Science (MFCS'03)}, {B.~Rovan} {and}
  {P.~Vojt{\'a}\v{s}}, Eds. LNCS, vol. 2747. Springer Berlin / Heidelberg,
  541--550.

\bibitem[\protect\citeauthoryear{Loyer and Straccia}{Loyer and
  Straccia}{2006}]{loyer:epistemic}
{\sc Loyer, Y.} {\sc and} {\sc Straccia, U.} 2006.
\newblock Epistemic foundation of stable model semantics.
\newblock {\em Journal of Theory and Practice of Logic Programming\/}~{\em 6},
  355--393.

\bibitem[\protect\citeauthoryear{Lukasiewicz}{Lukasiewicz}{1998}]{lukasiewicz-%
probabilistic}
{\sc Lukasiewicz, T.} 1998.
\newblock Probabilistic logic programming.
\newblock In {\em Proceedings of the 13th European Conference on Artificial
  Intelligence (ECAI'98)}. J. Wiley \& Sons, 388--392.

\bibitem[\protect\citeauthoryear{Lukasiewicz}{Lukasiewicz}{1999}]{Lukasiewicz-%
DisjunctiveProbabilisticLP}
{\sc Lukasiewicz, T.} 1999.
\newblock Many-valued disjunctive logic programs with probabilistic semantics.
\newblock In {\em Proceedings of the 5th International Conference on Logic
  Programming and Nonmonotonic Reasoning (LPNMR'99)}, {M.~Gelfond}, {N.~Leone},
  {and} {G.~Pfeifer}, Eds. LNCS, vol. 1730. Springer Berlin / Heidelberg,
  277--289.

\bibitem[\protect\citeauthoryear{Lukasiewicz}{Lukasiewicz}{2006}]{Luka06}
{\sc Lukasiewicz, T.} 2006.
\newblock Fuzzy description logic programs under the answer set semantics for
  the semantic web.
\newblock In {\em Proceedings of the Second International Conference on Rules
  and Rule Markup Languages for the Semantic Web (RuleML'06)}. 89--96.

\bibitem[\protect\citeauthoryear{Lukasiewicz and Straccia}{Lukasiewicz and
  Straccia}{2007a}]{LukaStraccia07}
{\sc Lukasiewicz, T.} {\sc and} {\sc Straccia, U.} 2007a.
\newblock Tightly integrated fuzzy description logic programs under the answer
  set semantics for the semantic web.
\newblock In {\em Proceedings of the First International Conference on Web
  Reasoning and Rule Systems (RR'07)}, {M.~Marchiori}, {J.~Pan}, {and}
  {C.~Marie}, Eds. LNCS, vol. 4524. Springer Berlin / Heidelberg, 289--298.

\bibitem[\protect\citeauthoryear{Lukasiewicz and Straccia}{Lukasiewicz and
  Straccia}{2007b}]{LukaStraccia-TopkRetrieval}
{\sc Lukasiewicz, T.} {\sc and} {\sc Straccia, U.} 2007b.
\newblock Top-k retrieval in description logic programs under vagueness for the
  semantic web.
\newblock In {\em Proceedings of the 1st international conference on Scalable
  Uncertainty Management (SUM'07)}, {H.~Prade} {and} {V.~Subrahmanian}, Eds.
  LNCS, vol. 4772. Springer Berlin / Heidelberg, 16--30.

\bibitem[\protect\citeauthoryear{Madrid and Ojeda-Aciego}{Madrid and
  Ojeda-Aciego}{2008}]{MadridOjeda-Aciego-2008a}
{\sc Madrid, N.} {\sc and} {\sc Ojeda-Aciego, M.} 2008.
\newblock Towards a fuzzy answer set semantics for residuated logic programs.
\newblock In {\em Proceedings of the 2008 IEEE/WIC/ACM International Conference
  on Web Intelligence and Intelligent Agent Technology (WI-IAT'08)}. 260--264.

\bibitem[\protect\citeauthoryear{Madrid and Ojeda-Aciego}{Madrid and
  Ojeda-Aciego}{2009}]{MadridAciego2009}
{\sc Madrid, N.} {\sc and} {\sc Ojeda-Aciego, M.} 2009.
\newblock On coherence and consistence in fuzzy answer set semantics for
  residuated logic programs.
\newblock In {\em Proceedings of the 8th International Workshop on Fuzzy Logic
  and Applications (WILF'09)}, {V.~Di~Ges{\`u}}, {S.~Pal}, {and}
  {A.~Petrosino}, Eds. LNCS, vol. 5571. Springer Berlin / Heidelberg, 60--67.

\bibitem[\protect\citeauthoryear{Madrid and Ojeda-Aciego}{Madrid and
  Ojeda-Aciego}{2011}]{madrid:existenceandunicity}
{\sc Madrid, N.} {\sc and} {\sc Ojeda-Aciego, M.} 2011.
\newblock On the existence and unicity of stable models in normal residuated
  logic programs.
\newblock {\em International Journal on Computer Mathematics\/}.
\newblock To Appear.

\bibitem[\protect\citeauthoryear{Nerode, Remmel, and Subrahmanian}{Nerode
  et~al\mbox{.}}{1997}]{NerodeRemmelSubrahmanian1997}
{\sc Nerode, A.}, {\sc Remmel, J.~B.}, {\sc and} {\sc Subrahmanian, V.~S.}
  1997.
\newblock Annotated nonmonotonic rule systems.
\newblock {\em Theoretical Computer Science\/}~{\em 171,\/}~1-2, 111--146.

\bibitem[\protect\citeauthoryear{Ng and Subrahmanian}{Ng and
  Subrahmanian}{1993}]{NgSubrahmanian-1993a}
{\sc Ng, R.} {\sc and} {\sc Subrahmanian, V.~S.} 1993.
\newblock A semantical framework for supporting subjective and conditional
  probabilities in deductive databases.
\newblock {\em Journal of Automated Reasoning\/}~{\em 10,\/}~2, 191--235.

\bibitem[\protect\citeauthoryear{Ng and Subrahmanian}{Ng and
  Subrahmanian}{1994}]{NgSubrahmanian-1994a}
{\sc Ng, R.} {\sc and} {\sc Subrahmanian, V.~S.} 1994.
\newblock Stable semantics for probabilistic deductive databases.
\newblock {\em Information and Computation\/}~{\em 110,\/}~1, 42--83.

\bibitem[\protect\citeauthoryear{Nicolas, Garcia, and St{\'e}phan}{Nicolas
  et~al\mbox{.}}{2005}]{nicolas:possibilistic-stable-models}
{\sc Nicolas, P.}, {\sc Garcia, L.}, {\sc and} {\sc St{\'e}phan, I.} 2005.
\newblock Possibilistic stable models.
\newblock In {\em Nonmonotonic Reasoning, Answer Set Programming and
  Constraints}. Dagstuhl Seminar Proceedings. Internationales Begegnungs- und
  Forschungszentrum f{\"u}r Informatik (IBFI).

\bibitem[\protect\citeauthoryear{Nicolas, Garcia, St\'{e}phan, and
  Lef\`{e}vre}{Nicolas et~al\mbox{.}}{2006}]{Nicolas:AMAI2006}
{\sc Nicolas, P.}, {\sc Garcia, L.}, {\sc St\'{e}phan, I.}, {\sc and} {\sc
  Lef\`{e}vre, C.} 2006.
\newblock Possibilistic uncertainty handling for answer set programming.
\newblock {\em Annals of Mathematics and Artificial Intelligence\/}~{\em
  47,\/}~1-2, 139--181.

\bibitem[\protect\citeauthoryear{Nov\'ak, Perfilieva, and
  Mo\u{c}ko\u{r}}{Nov\'ak et~al\mbox{.}}{1999}]{novak:1999}
{\sc Nov\'ak, V.}, {\sc Perfilieva, I.}, {\sc and} {\sc Mo\u{c}ko\u{r}, J.}
  1999.
\newblock {\em Mathematical Principles of Fuzzy Logic}.
\newblock Kluwer Academic Publishers.

\bibitem[\protect\citeauthoryear{Saad}{Saad}{2009a}]{Saad-2009}
{\sc Saad, E.} 2009a.
\newblock Extended fuzzy logic programs with fuzzy answer set semantics.
\newblock In {\em Proceedings of the 3rd International Conference on Scalable
  Uncertainty Management (SUM'09)}, {L.~Godo} {and} {A.~Pugliese}, Eds. LNCS,
  vol. 5785. Springer Berlin / Heidelberg, 223--239.

\bibitem[\protect\citeauthoryear{Saad}{Saad}{2009b}]{Saad:ECSQARU2009}
{\sc Saad, E.} 2009b.
\newblock Probabilistic reasoning by {SAT} solvers.
\newblock In {\em Proceedings of the 10th European Conference on Symbolic and
  Quantitative Approaches to Reasoning with Uncertainty (ECSQARU'09)},
  {C.~Sossai} {and} {G.~Chemello}, Eds. LNCS, vol. 5590. Springer Berlin /
  Heidelberg, 663--675.

\bibitem[\protect\citeauthoryear{Shapiro}{Shapiro}{1983}]{Shapiro1983}
{\sc Shapiro, E.~Y.} 1983.
\newblock Logic programs with uncertainties: a tool for implementing rule-based
  systems.
\newblock In {\em Proceedings of the Eighth international joint conference on
  Artificial intelligence (IJCAI'83)}, {A.~Bundy}, Ed. William Kaufmann,
  529--532.

\bibitem[\protect\citeauthoryear{Simons}{Simons}{2000}]{smodels}
{\sc Simons, P.} 2000.
\newblock Extending and implementing the stable model semantics.
\newblock Ph.D. thesis, Helsinki University of Technology.

\bibitem[\protect\citeauthoryear{Straccia}{Straccia}{2005}]{Straccia05queryans%
wering}
{\sc Straccia, U.} 2005.
\newblock Query answering in normal logic programs under uncertainty.
\newblock In {\em In 8th European Conferences on Symbolic and Quantitative
  Approaches to Reasoning with Uncertainty (ECSQARU-05)}, {L.~Godo}, Ed. LNCS,
  vol. 3571. Springer Berlin / Heidelberg, 470--470.

\bibitem[\protect\citeauthoryear{Straccia}{Straccia}{2006}]{straccia-annotated}
{\sc Straccia, U.} 2006.
\newblock Annotated answer set programming.
\newblock In {\em Proceedings of the 11th International Conference on
  Information Processing and Management of Uncertainty in Knowledge-Based
  Systems (IPMU'06)}.

\bibitem[\protect\citeauthoryear{Straccia}{Straccia}{2008}]{Straccia:reasoning%
web}
{\sc Straccia, U.} 2008.
\newblock Managing uncertainty and vagueness in description logics, logic
  programs and description logic programs.
\newblock In {\em Reasoning Web: 4th International Summer School 2008},
  {C.~Baroglio}, {P.~A. Bonatti}, {J.~M. uszynski}, {M.~Marchiori},
  {A.~Polleres}, {and} {S.~Schaffert}, Eds. LNCS, vol. 5224. 54--103.

\bibitem[\protect\citeauthoryear{Straccia, Ojeda-Aciego, and
  Dam\'{a}sio}{Straccia et~al\mbox{.}}{2009}]{straccia:fixedpoint}
{\sc Straccia, U.}, {\sc Ojeda-Aciego, M.}, {\sc and} {\sc Dam\'{a}sio, C.~V.}
  2009.
\newblock On fixed-points of multivalued functions on complete lattices and
  their application to generalized logic programs.
\newblock {\em SIAM Journal on Computing\/}~{\em 38,\/}~5, 1881--1911.

\bibitem[\protect\citeauthoryear{Subrahmanian}{Subrahmanian}{1994}]{Subrahmani%
an1994}
{\sc Subrahmanian, V.~S.} 1994.
\newblock Amalgamating knowledge bases.
\newblock {\em ACM Transactions on Database Systems\/}~{\em 19,\/}~2, 291--331.

\bibitem[\protect\citeauthoryear{Tarski}{Tarski}{1955}]{tarski:lattice}
{\sc Tarski, A.} 1955.
\newblock A lattice theoretical fixpoint theorem and its application.
\newblock {\em Pacific Journal of Mathematics\/}~{\em 5}, 285--309.

\bibitem[\protect\citeauthoryear{{Van Nieuwenborgh}, {De Cock}, and
  Vermeir}{{Van Nieuwenborgh} et~al\mbox{.}}{2007a}]{fasp1}
{\sc {Van Nieuwenborgh}, D.}, {\sc {De Cock}, M.}, {\sc and} {\sc Vermeir, D.}
  2007a.
\newblock Computing fuzzy answer sets using {DLVHEX}.
\newblock In {\em Proceedings of the 23rd International Conference on Logic
  Programming (ICLP'07)}, {V.~Dahl} {and} {I.~Niemel{\"a}}, Eds. LNCS, vol.
  4670. Springer Berlin / Heidelberg, 449--450.

\bibitem[\protect\citeauthoryear{{Van Nieuwenborgh}, {De Cock}, and
  Vermeir}{{Van Nieuwenborgh} et~al\mbox{.}}{2007b}]{FASP:amai}
{\sc {Van Nieuwenborgh}, D.}, {\sc {De Cock}, M.}, {\sc and} {\sc Vermeir, D.}
  2007b.
\newblock An introduction to fuzzy answer set programming.
\newblock {\em Annals of Mathematics and Artificial Intelligence\/}~{\em
  50,\/}~3-4, 363--388.

\bibitem[\protect\citeauthoryear{Vojt{\'a}s}{Vojt{\'a}s}{2001}]{Vojtas01}
{\sc Vojt{\'a}s, P.} 2001.
\newblock Fuzzy logic programming.
\newblock {\em Fuzzy Sets and Systems\/}~{\em 124,\/}~3, 361--370.

\bibitem[\protect\citeauthoryear{Wagner}{Wagner}{1998}]{Wagner-1998}
{\sc Wagner, G.} 1998.
\newblock Negation in fuzzy and possibilistic logic programs.
\newblock {\em Uncertainty Theory in Artificial Intelligence Series\/}~3,
  113--128.

\end{thebibliography}
